\documentclass{article}

\usepackage{arxiv}

\usepackage[utf8]{inputenc} % allow utf-8 input
\usepackage[T1]{fontenc}    % use 8-bit T1 fonts
\usepackage{hyperref}       % hyperlinks
\usepackage{url}            % simple URL typesetting
\usepackage{booktabs}       % professional-quality tables
\usepackage{amsfonts}       % blackboard math symbols
\usepackage{nicefrac}       % compact symbols for 1/2, etc.
\usepackage{microtype}      % microtypography
\usepackage{lipsum}
\usepackage{graphicx}
\graphicspath{ {./images/} }

\usepackage{amsmath}
\usepackage{amssymb}
\usepackage{amsthm}
\usepackage{listings}
\usepackage[table]{xcolor}

\definecolor{codegreen}{rgb}{0,0.6,0}
\definecolor{codegray}{rgb}{0.5,0.5,0.5}
\definecolor{codepurple}{rgb}{0.58,0,0.82}
\definecolor{backcolour}{rgb}{0.95,0.95,0.92}

\lstdefinestyle{mystyle}{
    backgroundcolor=\color{backcolour},   
    commentstyle=\color{codegreen},
    keywordstyle=\color{magenta},
    numberstyle=\tiny\color{codegray},
    stringstyle=\color{codepurple},
    basicstyle=\ttfamily\footnotesize,
    breakatwhitespace=false,         
    breaklines=true,                 
    captionpos=b,                    
    keepspaces=true,                 
    numbers=left,                    
    numbersep=5pt,                  
    showspaces=false,                
    showstringspaces=false,
    showtabs=false,                  
    tabsize=2
}

\newtheorem{theorem}{Theorem}[section]
\newtheorem{lemma}[theorem]{Lemma}
\newtheorem{corollary}[theorem]{Corollary}
\newtheorem{proposition}[theorem]{Proposition}
\theoremstyle{definition}
\newtheorem{remark}[theorem]{Remark}
\theoremstyle{definition}
\newtheorem{definition}[theorem]{Definition}

\title{Designing Rules for Choosing a Winner in a Debate}

\author{
 Alexander Heckett \\
  Carnegie Mellon University\\
  Pittsburgh, USA \\
  \texttt{aheckett@andrew.cmu.edu} \\
   \And
 Vincent Conitzer \\
  Carnegie Mellon University\\
  Pittsburgh, USA \\
  \texttt{conitzer@cs.cmu.edu}
}

\begin{document}
\maketitle
\begin{abstract}
We consider settings where an uninformed principal must hear arguments from two better-informed agents, corresponding to two possible courses of action that they argue for. The arguments are verifiable in the sense that the true state of the world restricts the arguments that can be made by the agents. Each agent simply wants to be chosen as the winner and does so strategically based on the rule set by the principal. How should the principal design the rule to choose the better action? We provide a formal framework for answering this question, exhibit some basic properties of it, study the computational problems of evaluating and optimizing the principal’s policy, and provide key error bounds.
\end{abstract}

% keywords can be removed
%\keywords{First keyword \and Second keyword \and More}

\section{Introduction}

In human life, to choose between two courses of action, or to form an opinion on whether something is true or false, a common way to proceed is to set up a {\em debate} between two entities, and then to declare an outcome based on the debate.  In some cases, the debaters inherently care about the outcome: for example, if we are deciding from which of two vendors to buy, each vendor will naturally be motivated to argue its own side.  But the setup is effective enough that often we {\em choose} to set up such an adversarial debate, asking one party to argue the one side and another to argue the other side, even if beforehand these parties had no relation to either side.  The idea is that the party on the right side of the debate will generally have access to more and better arguments, and thus if we set up the rules for determining the outcome of the debate right, then we are likely to make the right choice.  The idea of using debate to make better decisions in AI is an appealing one and has recently received attention especially in the context about keeping AI safe.  But how should we formalize this setup and thereby enable its use in the context of AI?  In this paper, we consider a probabilistic, Bayesian-game setup in which there is a true {\em scenario} based on which the debating agents have {\em actions} available, and we focus on the key problem of choosing a rule that determines the winner of the debate.

{\bf Running example.} Suppose a politician is considering whether a national minimum wage should be raised. The politician wants to determine whether increases in minimum wages have historically been associated with increased or decreased economic productivity. Let us suppose that changes in economic productivity can be objectively and deterministically evaluated. The politician could examine every nation or territory in recent history and compare the number of minimum wage increases associated with increased economic productivity against the number of minimum wage increases associated with decreased economic productivity. Unfortunately, the politician does not have the time to assess all of these historical precedents.

Instead, the politician's plan is to contact two economists, each of whom is a domain expert familiar with every minimum wage increase in every region in recent history. These economists are biased, however: one wishes for the national minimum wage to be raised while the other wishes to avoid such a raise. The politician wants to receive four examples from each of these economists.  The politician will then personally verify these eight historical examples and confirm that each example supports the side that submitted it. (If any example does not support the economist that submitted it, then that economist immediately loses, i.e., the politician will support the other economist's preferred decision on minimum wage.) 

It is understood that each economist can reference at least four supporting examples, so the mere existence of four examples on each side tells the politician nothing.  How, then, can this setup help the politician?
If one economist's argument is supported by many more historical examples than the other economist's argument, then this economist has more flexibility in how they respond to the politician.  Thus, the politician can set up a game that measures this flexibility, allowing the politician to gauge the relative number of examples that each economist had available. As an example of the possibilities available to the politician, note that the politician could specify a rule such as ``I will be more impressed by examples where the country in question has seven letters in its name.'' Even when the number of letters in the name of the country is completely unrelated to how inherently relevant the example is to the politician's own decision (and presumably it is irrelevant), such a rule will give a relative advantage to the economist with more examples on their side, who is more likely to be able to produce examples satisfying the criterion.
We examine how well the politician can do by constructing these kinds of games.

\subsection{Related Work}

{\em Mechanism design} concerns the design of rules that result in good outcomes when used by multiple agents with private information.  Famous general mechanisms include the VCG mechanism~\cite{Vickrey61,Clarke71:Multipart,Groves73:Incentives} and the Myerson auction~\cite{Myerson81:Optimal}; it is also possible to compute optimal mechanisms for the specific mechanism at hand, a process called {\em automated mechanism design}~\cite{Conitzer02:Mechanism,Conitzer04:Self}.  In most mechanism design settings studied in the literature, a result known as the {\em revelation principle} holds, which implies that we can restrict attention to mechanisms in which each agent simply reveals all its information and has no incentive to lie about this information.  However, the revelation principle can fail to hold in settings with {\em partial verification}, where the space of actions that an agent can take is restricted by the true state of the world.  Such is the case for the setting studied in this paper.  There are general characterizations for when the revelation principle does and does not hold in such settings (and more general ones in which the true state dictates a {\em cost} for taking each action)~\cite{Green86:Partially,Yu11:Mechanism,Kephart21:Revelation}; when it does hold, computing optimal mechanisms is generally easy, whereas otherwise it tends to be computationally hard~\cite{Auletta11:Alternatives,Kephart15:Complexity,Zhang21:Automated}.

{\em AI safety via debate}~\cite{Irving18:AI} is a framework for keeping highly competent AI systems safe.  When human beings must judge whether a course of action is safe and useful, there is a concern that they may not be capable of doing that accurately themselves.  The proposed solution is to have two AI agents debate which answer is right, with humans eventually choosing the winner of this zero-sum debate game. There has been significant theoretical progress on the kinds of languages which can be decided by such debate protocols~\cite{pmlr-v235-brown-cohen24a}. Our contributions in this paper can be seen as fitting in this broader agenda, although with more of an abstract framing.

Another related research direction is {\em argumentation theory}~\cite{Dung95:On}, and in particular  a mechanism design approach to that~\cite{RahwanMechanism}.  However, a key difference is that in that work, the relationship of which argument defeats which other ones is assumed to be given exogenously, and in the mechanism design work, the concern is the possible manipulation of an agent withholding arguments that it knows.

\subsection{Our Results}

We introduce three classes of debate games: common knowledge debate games, common knowledge distinguishing debate games (a special kind of common knowledge debate game), and private information distinguishing debate games. In all three classes, a principal (i.e., judge) is presented with two agents that both wish to win the principal's favor. The principal knows of a list of scenarios and the probability that each scenario occurs, though the principal does not expect to know the scenario during any debate. In each scenario, the two agents are restricted in what actions they can take and the principal knows these restrictions. Furthermore, in each scenario, the principal has desires over which agent should win (in our running example, the politician wishes to reward the economist that has more historical examples to call upon, which is determined by what actual historical examples of minimum wage increases exist, which the politician has some prior beliefs about but does not know). The principal commits to a policy, publicly and in advance, specifying for every pair of actions taken by the two agents what probability each agent has of winning (we assume that the agents only take a single action and that the agents move simultaneously, see Section 2 for commentary on this assumption). The principal attempts to minimize the error of their policy, defined to be the principal's expected loss in utility as a result of not being able to always favor their desired winning agent. We investigate the basic behavior of these three classes of debate games. Afterwards, we prove two main kinds of results: computational complexity results and error bounds.

Our computational complexity results establish that, for each of these three classes of debate games, evaluating the error of a specific policy can be done in polynomial time with respect to the description length of a game. However, in all three classes of debate games, the principal faces an NP-complete problem in determining whether it's even possible to construct a policy which rewards the principal's desired winner with probability $ 1 $, let alone in finding such a policy.

Afterwards, we use probabilistic methods to derive error bounds for two of our three classes of debate games; namely, both for common knowledge distinguishing debate games and for private information distinguishing debate games. For both classes of debate games, if in most scenarios the principal's desired winning agent has many more available actions than the principal's desired losing agent, then there exist policies for the principal which have very low error. The methods of constructing these policies differ greatly between the common knowledge and private information settings, however.

If the agents can see each others' available actions (the common knowledge case) then the principal can randomly choose subsets of ``highlighted'' actions for both agents in such a way that whichever agent is the desired loser has only a handful of available actions (if any) that are highlighted. The principal can then always favor highlighted actions over non-highlighted actions and can randomly assign winners in highlighted-highlighted action pairs. Making such random choices makes it very likely that the desired winner has some available highlighted action which defeats each of the desired loser's handful of available highlighted actions, thus giving the policy a very low error.

If the agents cannot see each others' available actions (the private information case), then the principal can instead announce a random ordering of the combined action sets of the two agents and can always favor whichever agent plays a higher-ranked action. Such policies do not depend on the agents knowing each others' available actions since the agents never benefit from playing anything other than their highest-ranked action. If the desired winning agent has many more available actions than the desired losing agent, then the desired winning agent is proportionally more likely to have the highest ranked action so the error is again small.

We use these two techniques to establish upper bounds on the errors that the principal can achieve. For common knowledge distinguishing debate games we show that the minimum possible policy error decays faster than the reciprocal of any polynomial in the ratio of the minimum number of actions available to the desired winner to the maximum number of actions available to the desired loser. By contrast, our error bound for private information distinguishing debate games is weaker, decaying asymptotically as the reciprocal of this ratio. We prove that such asymptotic behavior is the best we can achieve in the general private information setting by constructing an example game whose minimum error has this asymptotic behavior.

\section{Setups}

We will assume that the agents each take one action and that these actions are taken simultaneously, though this choice should not be seen as restricting the kinds of games that our framework models. For example, we can take extensive-form representations of debates (which could contain multiple rounds of interaction between the agents) and convert them into normal form and then apply our formalism (at the expense of potentially dramatically increasing the size of the representation).

Throughout the paper, we assume there is a special ``default'' action $\delta$ that is always available to both agents (which could be thought of as not doing anything), to avoid the action set ever being empty. For notational convenience later on, we keep this action outside of the set of (other) actions $A_i$ for an agent, so that an agent's {\em full} set of actions is always $A_i \cup \{\delta\}$.

We first introduce a general setup under which the actions available to each agent are common knowledge between the agents:

\begin{definition}
\label{ckdg}
A \textit{common knowledge debate game} (abbreviated CKDG) is a tuple $ (A_1, A_2, S, P, C_1, C_2, u) $ where $ A_i $ is a finite set of actions potentially available to agent $ i \in \{1, 2\} $ (we require $ \delta \notin A_i $), $ S $ is a finite set of scenarios, $ P : S \to [0, 1] $ is a probability mass function over scenarios (which implies $ \sum_{s \in S} P(s) = 1 $), $ C_i : S \to \mathcal{P}(A_i) $ tells us what actions agent $ i \in \{1, 2\} $ can choose in any given scenario, and $ u : \{1, 2\} \times S \to \mathbb{R} $ tells us the principal's utility if a given agent wins in a given scenario.
\end{definition}

A common knowledge debate game specifies the information the principal possesses when trying to decide what policy to commit to publicly. Let's now define what such a policy consists of:

\begin{definition}
If $ B = (A_1, A_2, S, P, C_1, C_2, u) $ is a CKDG then a \textit{$ B $-policy} (or simply \textit{policy}) for the principal is a map $ M : \{1,2\} \times (A_1 \cup \{\delta\}) \times (A_2 \cup \{\delta\}) \to [0, 1] $ which specifies the probability that the principal will favor one of the agents given the actions performed by both agents. We require $ M(1,a_1,a_2) + M(2,a_1,a_2) = 1 $ for all $ a_1 \in A_1 \cup \{ \delta \} $ and all $ a_2 \in A_2 \cup \{ \delta \} $.
\end{definition}

Note that we will assume throughout that the principal can commit to whatever policy they wish. The section on future research contains more commentary on this assumption.

Now that we've established what a policy consists of, we can examine what makes a policy good or bad. We begin by understanding how the agents respond to a given policy:

\begin{definition}
Let $ B $ be a CKDG and let $ M $ be a $ B $-policy. Define $ w^i_{B,M} : S \to [0, 1] $ to be the function which computes the probability that agent $ i \in \{1, 2\} $ will win in any given scenario.  This is well defined because, due to the common-knowledge assumption, it corresponds to a zero-sum game the agents play: if $ s \in S $ then $ w^i_{B,M}(s) $ is the value to agent $ i $ of the two-player zero-sum normal-form game with payoff matrix $ M(i, \cdot, \cdot) $, where the first unspecified argument must come from $ C_1(s) \cup \{ \delta \} $ and where the second unspecified argument must come from $ C_2(s) \cup \{ \delta \} $.
\end{definition}

We are now ready to characterize how bad a given policy is:

\begin{definition}
Let $ B $ be a CKDG. The \textit{error} of a $ B $-policy $ M $, $ e_B^{\text{CKDG}}(M) $, is the utility lost to the principal by being unable to favor the desired winner:
\begin{align*}
e^{\text{CKDG}}_B(M) & = \underset{s \sim P}{\mathbb{E}} \begin{cases} d_B(s) w_{B,M}^1(s) & u(1,s) \leq u(2,s) \\ d_B(s) w_{B,M}^2(s) & u(1,s) > u(2,s) \end{cases}
\end{align*}
where $ d_B(s) = |u(1,s) - u(2,s)| $ is the difference in utility between the two agents winning a given scenario.
\end{definition}

CKDG's allow for the two debating agents to appear {\em ex ante} different. But we are especially interested in the situation where the two players appear {\em ex ante} similar but have two different roles (say, the role of the economist with more historical examples or the role of the economist with fewer historical examples) and our task is to distinguish which player occupies which role (in our running example, our goal is to determine which economist's argument is supported by more historical precedents). To model this specific class of problems, we introduce the \textit{distinguishing} variant of common knowledge debate games:

\begin{definition}
\label{ckddg}
A \textit{common knowledge distinguishing debate game} (abbreviated CKDDG) is a tuple $ (A, S, P, C_w, C_l) $ where $ A $ is a finite set of actions (we require $ \delta \notin A $), $ S $ is a finite set of anonymized scenarios, $ P : S \to [0, 1] $ is a probability mass function over anonymized scenarios, $ C_w : S \to \mathcal{P}(A) $ tells us what actions are available to the agent that we want to win in a given anonymized scenario, and $ C_l : S \to \mathcal{P}(A) $ tells us what actions are available to the agent that we want to lose in a given anonymized scenario. 
\end{definition}

Note that unlike in general CKDGs, ``the agent that we want to win in a given anonymized scenario'' and ``the agent that we want to lose in a given anonymized scenario'' no longer refer to either agent in particular. Instead these refer to roles that (according to our prior) either agent will occupy with 50\% probability. More formally, to see how we interpret common knowledge distinguishing debate games as common knowledge debate games:

\begin{definition}
If $ B = (A,S,P,C_w,C_l) $ is a CKDDG then \textit{the CKDG induced by $ B $} is $ B' = (A_1', A_2', S', P', C_1', C_2', u') $ where $ A_1' = A_2' = A $, $ S' = \{1, 2\} \times S $ (a scenario in the CKDG is an anonymized scenario in the CKDDG together with a labeling of which agent should win), $ P'((i, s)) = \frac{1}{2} P(s) $ (for every anonymized scenario in the distinguishing problem we uniformly randomly choose which agent is which), and $ C_j'((i,s)) = \begin{cases} C_w(s) & i = j \\ C_l(s) & i \neq j \end{cases} $, $ u(j,(i,s)) = \begin{cases} 1 & i = j \\ 0 & i \neq j \end{cases} $ for $ j \in \{1, 2\}$. 
\end{definition}

When the context is clear we may refer to anonymized scenarios simply as scenarios. 

Since we can always unfurl a CKDDG into a CKDG, given a CKDDG we can talk about policies and their errors:
\begin{definition}
If $ B $ is a CKDDG then a \textit{$ B $-policy} (or simply \textit{policy}) is a $ B' $-policy where $ B' $ is the CKDG induced by $ B $.
\end{definition}
\begin{definition}
If $ B $ is a CKDDG, $ M $ is a $ B $-policy, and $ B' $ is the CKDG induced by $ B $, then the \textit{error} of $ M $ is $ e^{\text{CKDDG}}_B(M) = e^{\text{CKDG}}_{B'}(M) $. 
\end{definition}
More explicitly, if $ B = (A,S,P,C_w,C_l) $ is a CKDDG then a $ B $-policy $ M $ is a map from $ \{1, 2\} \times (A \cup \{\delta\})^2 $ to $ [0, 1] $ such that paired probabilities sum to $ 1 $. The error of a CKDDG policy has a simple interpretation: it is the probability that the desired loser actually wins:
\begin{remark}
If $ B $ is a CKDDG and if $ M $ is a $ B $-policy then:
\begin{align*}
e_{B}^{\text{CKDDG}}(M) & = \underset{s \sim P}{\mathbb{E}} \left [ \frac{w^1_{B',M}((2,s)) + w^2_{B',M}((1,s))}{2} \right ]
\end{align*}
\end{remark}

We can now model our running economic example as a CKDDG. Since the economists are domain experts, they are expected to know all instances of minimum wage increases that both they and their opponent can reference; so, the common knowledge assumption is reasonable. Let $ H $ denote the set of historical periods (for specific countries) that both economists can point to; for example, ``USA, 1980-2000'' or ``Japan, 1960-1980.'' Let us suppose that the politician recognizes 200 such historical periods. The actions potentially available to our economists, then, are represented by selections of four historical periods: $ A = \begin{pmatrix} H \\ 4 \end{pmatrix} $. Our politician will thus be expected to supply, implicitly or explicitly, a roughly 65-million-by-65-million payoff matrix to the economists. As for the scenarios, suppose the politician believes that there are 40 instances of minimum wage increases among these 200 historical periods, and suppose the politician believes that either 20\% of these increases were associated with improved economic productivity or 20\% of these increases were associated with decreased economic productivity. In other words, suppose the politician believes that one side can call upon 32 pieces of evidence while the other side can call upon 8 pieces of evidence. Suppose the politician considers both regimes equally likely and considers any assignment of these outcomes to historical periods equally likely. Then we can take our anonymized scenarios to represent ways of choosing 32 pieces of evidence for one side (the desired winner) and 8 pieces of evidence for the other side (the desired loser): $ S = \left \{ (H_w,H_l) \in \begin{pmatrix} H \\ 32 \end{pmatrix} \times \begin{pmatrix} H \\ 8 \end{pmatrix} : H_w \cap H_l = \emptyset \right \} $. When this CKDDG is unfurled into a CKDG, each anonymized scenario will be unfurled into two de-anonymized scenarios: one de-anonymized scenario in which the 32 pieces of evidence show higher economic productivity and the 8 pieces of evidence show lower productivity and another de-anonymized scenario in which the 32 pieces of evidence show lower productivity and the 8 pieces of evidence show higher productivity. Since the politician's prior over scenarios is uniform, $ P(s) = \frac{1}{|S|} $ for all $ s \in S $. Finally, in any scenario, the desired winner can mention any subset of four pieces of evidence supporting the desired winner and the desired loser can mention any subset of four pieces of evidence supporting the desired loser: $ C_w((H_w,H_l)) = \begin{pmatrix} H_w \\ 4 \end{pmatrix} \subseteq A, C_l((H_w,H_l)) = \begin{pmatrix} H_l \\ 4 \end{pmatrix} \subseteq A $. With these definitions, $ (A, S, P, C_w, C_l) $ specifies a CKDDG modeling our setup.

We will now introduce another class of debate games, one in which the agents do not necessarily know what actions are available to each other. We will also jump directly to the distinguishing version of the problem since the error of a policy is not well-defined for general non-distinguishing private-information debate games\footnote{We defined the error of a CKDG as the loss of utility induced by being unable to choose the preferred winner on a scenario-by-scenario basis. In the private-information setting, it's possible to have multiple equilibria of the Bayesian game between agents 1 and 2 which result in different losses in utility to the principal.}.
\begin{definition}
\label{piddg}
A \textit{private information distinguishing debate game} (abbreviated PIDDG) is a tuple $ (A, S, P, C_w, C_l) $ subject to the exact same constraints as a CKDDG.
\end{definition}

\begin{definition}
If $ B $ is a PIDDG then a \textit{$ B $-policy} (or simply \textit{policy}) $ M $ is a map from $ \{1,2\} \times (A \cup \{\delta\})^2 $ to $ [0, 1] $ satisfying $ M(1,a_1,a_2) + M(2,a_1,a_2) = 1 $ for all $ a_1, a_2 \in A \cup \{\delta\} $.
\end{definition}

So far, the definition of PIDDG's has been identical to the definition of CKDDG's. The difference between PIDDG's and CKDDG's lies in the way we define the error of a policy. In particular, while a CKDDG gets decomposed into an even number of normal-form games, a PIDDG gets decomposed into two more intricate Bayesian games:
\begin{definition}
Let $ B = (A, S, P, C_w, C_l) $ be a PIDDG and let $ M $ be a $ B $-policy. Define $ G_1(B,M) $ to be the Bayesian game where agent $ 1 $ is the desired winner. Formally, $ G_1(B,M) $ has agent set $ \{1, 2\} $, for $ i \in \{1, 2\} $ agent $ i $ has action set $ A \cup \{\delta\} $ and type set $ \mathcal{P}(A) $ (an agent's type represents the non-default actions they can access). For $ a_1, a_2 \in A \cup \{\delta\} $, the utility to agent $ i $ of agent $ 1 $ playing $ a_1 $ while having type $ t_1 $ and agent $ 2 $ playing $ a_2 $ while having type $ t_2 $ is $ M(i,a_1,a_2) $ provided that $ a_1 \in t_1 \cup \{\delta\}, a_2 \in t_2 \cup \{\delta\} $. If agent $ i $ tries to play an unavailable action ($ a_i \notin t_i \cup \{\delta\} $) but agent $ -i $ plays an available action then agent $ i $ gets utility $ 0 $ and agent $ -i $ gets utility $ 1 $. If both agents play unavailable actions then both agents get $ \frac{1}{2} $ utility. Finally, the joint probability distribution over types is induced by $ P $ where in scenario $ s \in S $ agent $ 1 $ gets type $ C_w(s) $ and agent $ 2 $ gets type $ C_l(s) $.
\end{definition}
\begin{definition}
Let $ B = (A,S,P,C_w,C_l) $ be a PIDDG and let $ M $ be a $ B $-policy. Define $ G_2(B,M) $ identically to $ G_1(B,M) $ except that in scenario $ s \in S $, agent $ 1 $ gets type $ C_l(s) $ and agent $ 2 $ gets type $ C_w(s) $.
\end{definition}

We can again define the error as the probability that the principal's desired loser actually wins. In particular, the error is the average of the probability that agent 1 wins when in the Bayesian game where agent 2 should win and the probability that agent 2 wins when in the Bayesian game where agent 1 should win:

\begin{definition}
If $ B $ is a PIDDG and if $ M $ is a $ B $-policy then the \textit{error} of $ M $ is:
\begin{align*}
e^{\text{PIDDG}}_B(M) & = \frac{v_1(G_2(B,M)) + v_2(G_1(B,M))}{2}
\end{align*}
where $ v_i(G_j(B,M)) $ denotes the value of the game $ G_j(B,M) $ to agent $ i $ for $ i, j \in \{1, 2\} $.
\end{definition}

Note that the agents know whether they're in the game where they're the desired winner. This assumption matters because knowing whether their opponent is the desired winner or desired loser changes the distribution over available actions their opponent is likely to possess. 

\section{Initial Observations}

We first present some simple properties of debate games.

\begin{remark}
\label{injective}
If the action availability functions $ C_w, C_l $ in a PIDDG are injective, then the agents have enough information to determine the scenario and the PIDDG functions identically to a CKDDG.
\end{remark}

\begin{remark}
\label{minimum_error}
Each of the error functions is continuous in the policy, and the space of policies is always compact, so for every instance of any of our three classes of debate games, there is always some policy with the minimum possible error.
\end{remark}

Let's now consider how changes to the setup impact the agents.

\begin{remark}
\label{policy_monotonicity}
In any context, if the principal modifies the policy $ M $ by increasing $ M(1, \cdot, \cdot) $ element-wise at the expense of $ M(2, \cdot, \cdot) $, then agent $ 1 $ cannot be worse off and agent $ 2 $ cannot be better off. This is apparent in the CKDG/CKDDG context because the normal-form game associated with each scenario shifts in favor of agent $ 1 $. In the PIDDG context, changes to $ M $ propagate monotonically through the two associated Bayesian games and thus monotonically impact the values of those games. The symmetric result holds if the principal lowers $ M(1, \cdot, \cdot) $ and raises $ M(2, \cdot, \cdot) $.
\end{remark}

The situation is more complicated if we try to change action availability:

\begin{remark}
\label{sk_action_monotonicity}
For a CKDG, if we fix a policy and then change a scenario to give an agent more actions, that agent cannot end up worse off (this follows from the same property being true of the normal-form games associated with each scenario). Likewise, for a CKDDG, giving the desired winner more actions in an anonymous scenario cannot increase the error and giving the desired loser more actions cannot decrease it.
\end{remark}

This monotonicity fails in the private information case, however, because giving an agent more actions could hinder that agent's ability to infer what actions are available to its opponent.

\begin{proposition}
\label{piddg_monotonicity_failure}
There exist PIDDG's $ B = (A, S, P, C_w, C_l) $ and $ B' = (A, S, P, C_w', C_l') $ with $ C_w'(s) \supseteq C_w(s) $ and $ C_l'(s) \subseteq C_l(s) $ for all $ s \in S $ such that the minimum possible error of a policy for $ B' $ exceeds the minimum possible error of a policy for $ B $.
\end{proposition}
\begin{proof}
Let $ A = S = \{0, 1, 2\} $, let $ P $ be uniform over $ S $, and define $ C_w, C_l, C_w', C_l' $ by $ C_w(s) = \left \{ (s+1) \mod 3 \right \} $, $ C_w'(s) = A $, $ C_l(s) = C_l'(s) = \{s\} $. Note that $ B $ encodes rock-paper-scissors where the agents can either play an assigned action or take the default action. There is a policy for $ B $ with zero error (where we simply copy the reward-structure from rock-paper-scissors while punishing any submission of the default action) but the minimum error for $ B' $ is strictly positive (the desired winner is now blind to the desired loser's available actions and no policy can allow agent $ 1 $ an action which beats every action of agent $ 2 $ and agent $ 2 $ an action which beats every action of agent $ 1 $). See appendix A.1 for details.
\end{proof}

Let's now consider the role of randomization in debate games. Our definition of a policy allows agents to win with some probability for every pair of actions. We may ask whether it's ever beneficial to avoid declaring an absolute winner for some pair of actions, and the answer is yes:
\begin{proposition}
\label{non_binary_policy}
There exist CKDG's, CKDDG's, and PIDDG's such that the minimum error of any policy is strictly less than the minimum error among policies whose image is contained in $ \{0, 1\} $.
\end{proposition}
\begin{proof}
Consider the CKDDG $ (A, S, P, C_w, C_l) $ defined by $ A = \{0, 1, 2, 3\} $, $ S = \{0, 1, 2, 3, 4\} $, $ P(0) = P(1) = P(2) = \frac{3}{10} $, $ P(3) = \frac{6}{100} $, $ P(4) = \frac{4}{100} $, $ C_w(i) = \{(i+1) \mod 3\} $ for $ i \in \{0, 1, 2\} $, $ C_w(3) = \{0, 1, 2\} $, $ C_w(4) = \{3\} $, $ C_l(i) = \{i\} $ for $ i \in \{0, 1, 2\} $, $ C_l(3) = \{0, 1, 2, 3\} $, $ C_l(4) = \{0, 1, 2\} $. Then there exists a policy whose error is strictly less than the minimum error among policies which map to $ \{0, 1\} $. For details, see appendix A.2. 

We can unfurl our CKDDG to a CKDG. We also have injective action availability functions so our CKDDG behaves identically to a PIDDG. Thus our result carries over to CKDG's and PIDDG's.
\end{proof}

Not only does it sometimes benefit the principal to choose a probabilistic policy, it also sometimes benefits the principal to choose a policy that incentivizes the agents to act probabilistically as well:
\begin{proposition}
\label{mixed_agent_strategies}
There exist CKDG's, CKDDG's, and PIDDG's such that the minimum error of any policy is strictly less than the minimum error among policies such that the agents use pure strategies in equilibrium. 
\end{proposition}
\begin{proof}
Consider the CKDDG $ (A, S, P, C_w, C_l) $ where $ A = S = \{0, 1, 2, 3\} $, $ P(0) = P(1) = P(2) = \frac{3}{10} $, $ P(3) = \frac{1}{10} $, $ C_w(i) = \{(i+1) \mod 3\} $ for $ i \in \{0, 1, 2\} $, $ C_w(3) = \{0, 1, 2\} $, $ C_l(i) = \{i\} $ for $ i \in \{0, 1, 2\} $, and $ C_l(3) = \{0, 1, 2, 3\} $. Then there is a policy whose error is strictly less than the minimum error among policies which incentivize pure strategies. For details, see appendix A.3.

This CKDDG can be unfurled into a CKDG and it has injective action availability functions so it behaves identically to a PIDDG. Thus we get the same result for CKDG's and PIDDG's.
\end{proof}

\section{Complexity Results}

We begin by showing that we can find the error of any individual policy for any kind of debate game in polynomial time with respect to the size of the game description:

\begin{proposition}
\label{ckdg_eval}
Given a CKDG $ B = (A_1, A_2, S, P, C_1, C_2, u) $ and a policy $ M $, we can compute $ e^{\text{CKDG}}_B(M) $ in polynomial time.
\end{proposition}
\begin{proof}
If we look at the formula for $ e^{\text{CKDG}}_B(M) $, we see that since we can easily access $ u $ and by extension $ d_B $ it suffices to compute $ w^{i_s}_{B,M}(s) $ for every $ s \in S $ and for some $ i_s \in \{1, 2\} $. Computing $ w^{i_s}_{B,M}(s) $ involves finding the value of a two-player zero-sum normal-form game with $ |C_1(s)| + 1 $ actions for agent $ 1 $ and $ |C_2(s)| + 1 $ actions for agent $ 2 $. We can thus compute $ e^{\text{CKDG}}_B(M) $ by finding the values of a polynomial number of polynomial-sized two-player zero-sum normal-form games, so we can find $ e^{\text{CKDG}}_B(M) $ in polynomial time.
\end{proof}

\begin{proposition}
\label{ckddg_eval}
Given a CKDDG $ B = (A, S, P, C_w, C_l) $ and a policy $ M $, we can compute $ e^{\text{CKDDG}}_B(M) $ in polynomial time.
\end{proposition}
\begin{proof}
If $ B' $ is the CKDG induced by $ B $ then the size of $ B' $ is polynomial (indeed linear) in the size of $ B $ so $ e^{\text{CKDDG}}_B(M) = e^{\text{CKDG}}_{B'}(M) $ can be computed in polynomial time.
\end{proof}

\begin{proposition}
\label{piddg_eval}
Given a PIDDG $ B = (A, S, P, C_w, C_l) $ and a policy $ M $, we can compute $ e^{\text{PIDDG}}_B(M) $ in polynomial time.
\end{proposition}
\begin{proof}
$ e^{\text{PIDDG}}_B(M) $ is the average of the values of two two-player zero-sum Bayesian games. These two-player zero-sum Bayesian games both have description lengths polynomial in the description lengths of $ B, M $ so we can compute their values in polynomial time using two polynomial-sized linear programs.
\end{proof}

While evaluating the error of an individual policy is important, we are more likely to care about optimizing over policies to minimize the error. Unfortunately, for this kind of problem we derive several hardness results. In particular, we prove that it is NP-complete to determine whether a CKDG, CKDDG, or PIDDG has a policy with zero error (hereafter called a \textit{perfect policy}). Note that this hardness result extends to computing the minimum possible error and finding a minimizing policy in the natural way. 

Let's start by considering CKDG's with injective action availability functions (we add this injectivity constraint to aid the jump from CKDDG hardness to PIDDG hardness later on):
\begin{lemma}
\label{ckdg_hardness}
It's NP-hard to determine whether a CKDG with injective action availability functions has a perfect policy.
\end{lemma}
\begin{proof}
Let's reduce from SAT. Suppose $ x_1, ..., x_n $ are variables and suppose that for every $ i \in \{1, ..., m\} $, $ K_i = \underset{j \in J_i}{\bigvee} (n_{i,j} \oplus x_j) $ is a clause where $ J_i \subseteq \{1,...,n\} $ and $ n_{i,j} \in \{\bot, \top\} $ for all $ j \in J_i $. Without loss of generality we can assume that there is at least one clause, no two clauses are identical, and each clause references at least two distinct variables. Let's build a CKDG $ (A_1, A_2, S, P, C_1, C_2, u) $ with $ A_1 = \{1, ..., n\} \times \{\bot,\top\} $, $ A_2 = \{1, ..., m\} \times \{1, ..., n\} $, $ S = \left ( \{1\} \times \{1,...,n\} \right ) \cup \left ( \{2\} \times \{1, ..., m\} \right ) $, and $ P $ uniform over $ S $. We can let $ C_1((1,k)) = \{(k,\bot),(k,\top)\} $ and $ C_2((1,k)) = \{1, ..., m\} \times \{k\} $. We can let $ C_1((2,k)) = \{(j, \neg n_{k,j}) : j \in J_k\} $ and $ C_2((2,k)) = \{k\} \times J_k $. Finally, we can let $ u(p,(w,k)) $ be $ 1 $ if $ p = w $ and $ 0 $ otherwise.

This CKDG has a perfect policy if and only if $ K_1 \land ... \land K_m $ is satisfiable. Intuitively, if $ x_1, ..., x_n $ is a satisfying assignment for $ K_1 \land ... \land K_m $ then the policy $ M(1,a_1,a_2) = \begin{cases} 1 & \exists i,j[a_1 = (j,x_j) \land a_2 = (i,j)] \\ 0 & \text{otherwise} \end{cases} $ is perfect (if the row player is the desired favorite in a variable-indexed scenario then they can play their variable and its value, if the column player is the desired favorite in a clause-indexed scenario then they can play their clause and the variable whose literal makes the clause satisfied). Furthermore, any perfect policy must closely resemble this construction. For a detailed proof, see appendix B.1. Furthermore, $ C_1 $ and $ C_2 $ are injective as is also argued in appendix B.1.
\end{proof}

Now we can show the same result for CKDDG's:

\begin{lemma}
\label{ckddg_hardness}
It's NP-hard to determine whether a CKDDG with injective action availability functions has a perfect policy.
\end{lemma}
\begin{proof}
Let's reduce from determining whether a CKDG with injective action availability functions has a perfect policy (which we know to be NP-hard by Lemma \ref{ckdg_hardness}). Suppose $ (A_1, A_2, S, P, C_1, C_2, u) $ is a CKDG where $ C_1, C_2 $ are injective. Consider the set of scenarios which sometimes occur and on which we are not indifferent, $ S' = \left \{ s' \in S : P(s') > 0 \land u(1,s') \neq u(2,s') \right \} $. If $ S' = \emptyset $ then every policy for our original CKDG is perfect. Otherwise, consider the CKDDG $ (A, S', P', C_w, C_l) $ where $ A = (\{1\} \times (A_1 \cup \{\delta\})) \cup (\{2\} \times (A_2 \cup \{\delta\})) $, $ P' $ is uniform over $ S' $, and for $ s' \in S' $ if $ u(1,s') > u(2,s') $ then $ C_w(s') = \{1\} \times (C_1(s') \cup \{\delta\}), C_l(s') = \{2\} \times (C_2(s') \cup \{\delta\}) $ but if $ u(1,s') < u(2,s') $ then $ C_w(s') = \{2\} \times (C_2(s') \cup \{\delta\}), C_l(s') = \{1\} \times (C_1(s') \cup \{\delta\}) $. There is a perfect policy for this new CKDDG if and only if there is a perfect policy for the original CKDG. Furthermore, $ C_w $ and $ C_l $ are injective. For proofs of these claims, see appendix B.2.
\end{proof}

Finally, we are ready to prove NP-completeness for all of our problems:

\begin{theorem}
\label{everything_completeness}
It's NP-complete to determine whether a CKDG, CKDDG, or PIDDG has a perfect policy.
\end{theorem}
\begin{proof}
Lemma \ref{ckddg_hardness} tells us that it's NP-hard to tell whether a CKDDG with injective action availability functions has a perfect policy. However, Remark \ref{injective} tells us that such a CKDDG behaves identically to a PIDDG (as the available actions reveal everything). Thus it's NP-hard to tell whether a PIDDG with injective action availability functions has a perfect policy. As a result, we have special cases of CKDDG's and PIDDG's where it's NP-hard to tell whether perfect policies exist, so it's also hard for CKDG's which generalize CKDDG's.

Finally, policies can be expressed in polynomial space so Propositions \ref{ckdg_eval}, \ref{ckddg_eval}, and \ref{piddg_eval} show that evaluating the error of a given policy can be done in polynomial time, so we have membership in NP as well.
\end{proof}

\section{Error Bounds}

In spite of the hardness results, finding policies with extremely low error is often feasible if we expect the agent we want to win to possess many more actions than the agent we want to lose. 

Let's begin with the CKDDG case. We will consider policies which randomly highlight actions for both agents and which force agents to play available actions that are highlighted. We will consider randomly assigning a winner to each pair of highlighted actions. To see why this technique is promising, suppose that in some scenario the desired loser has access to $ 100 $ actions and the desired winner has access to $ 10,000 $ actions. Suppose we highlight actions with probability $ 0.01 $ and suppose that our desired loser ends up with $ 1 $ highlighted action and the desired winner ends up with $ 100 $ highlighted actions. Then if our policy favors the desired winner on any of the $ 100 $ intersection points between highlighted available actions, the desired winner can guarantee a victory. In this case, the probability that our policy would favor the desired loser is only $ 2^{-100} $. Note that this reasoning requires that the desired winner knows both the policy and the desired loser's available actions.

\begin{theorem}
\label{ckddg_upper_bound}
Suppose that $ (A, S, P, C_w, C_l) $ is a CKDDG and suppose that $ m, n \in \mathbb{N}^{+}, \varepsilon \in [0, 1] $ are such that if $ s $ is sampled from $ S $ according to $ P $ then the probability that $ |C_w(s)| < m $ or that $ |C_l(s)| > n $ is at most $ \varepsilon $. Suppose $ m > n $ and define $ r = \frac{m}{n} $. Let $ \alpha \in (0, 1] $ and let $ k, l \in \mathbb{N}^+ $. Let $ \left ( \mu_i \right )_{i = 1}^{k} $ be a non-increasing sequence of elements of $ (0, 1] $ and let $ \left ( \nu_j \right )_{j = 1}^{l} $ be a non-decreasing sequence of elements of $ \left [ \frac{\alpha}{\log_2(r)}, \infty \right ) $. For every $ i \in \{1, ..., k\}, j \in \{1, ..., l\} $ define $ \theta_{i,j} = e^{-\mu_i \alpha r^{1 - \nu_j}} $, $ \zeta_i = e^{-\alpha r} \left ( \frac{e}{\mu_i} \right )^{\mu_i \alpha r} $, and $ \xi_j = e^{-\alpha} r^{-\nu_j \log_2 \left ( \frac{\nu_j \log_2(r)}{e \alpha} \right )} $. Further define $ \zeta_0 = \xi_0 = 1 $. Then there exists a policy with error at most $ \sum_{i = 1}^{k} \sum_{j = 1}^{l} \theta_{i,j} \zeta_{i-1} \xi_{j-1} + \zeta_{k} + \xi_{l} + \varepsilon $.
\end{theorem}

For example, in our running economic example, we have $ |C_w(s)| = \begin{pmatrix} 32 \\ 4 \end{pmatrix} $ and $ |C_l(s)| = \begin{pmatrix} 8 \\ 4 \end{pmatrix} $ for every anonymized scenario $ s \in S $. We can thus choose $ r = \frac{\begin{pmatrix} 32 \\ 4 \end{pmatrix}}{\begin{pmatrix} 8 \\ 4 \end{pmatrix}} = \frac{3596}{7} $ and $ \varepsilon = 0 $ in the theorem statement above. If we consider $ \alpha = 0.315392 $ then we can find valid $ k,l,(\mu_i)_{i=1}^k,(\nu_j)_{j=1}^{l} $ such that $ \sum_{i = 1}^{k} \sum_{j = 1}^{l} \theta_{i,j} \zeta_{i-1} \xi_{j-1} + \zeta_k + \xi_l \leq 0.000035 $, meaning that our politician can choose the economist who was supported by more historical precedent with at least $ 99.99\% $ accuracy. For details, see appendix C. According to our following proof construction, the politician will sample a policy by independently highlighting each of the $ 2 \cdot \begin{pmatrix} 200 \\ 4 \end{pmatrix} = 129369900 $ player-action pairs with probability $ \frac{\alpha}{\begin{pmatrix} 8 \\ 4 \end{pmatrix}} = 0.0045056 $.

\begin{proof}

\begin{table}
\centering
\def\arraystretch{1.25}
\begin{tabular}{c|c|c|c|c|c|c|c|}
& a & b & c & d & e & f & $ \delta $ \\ \hline
a & \cellcolor{gray!50}$0$ & $ \frac{1}{2} $ & \cellcolor{gray!50}$0$ & \cellcolor{gray!50}$0$ & $\frac{1}{2}$ & $\frac{1}{2}$ & $\frac{1}{2}$ \\ \hline
b & \cellcolor{gray!100}$1$ & \cellcolor{gray!50}$1$ & \cellcolor{gray!100}$0$ & \cellcolor{gray!100}$1$ & \cellcolor{gray!50}$1$ & \cellcolor{gray!50}$1$ & \cellcolor{gray!50}$1$ \\ \hline
c & \cellcolor{gray!50}$0$ & $\frac{1}{2}$ & \cellcolor{gray!50}$0$ & \cellcolor{gray!50}$0$ & $\frac{1}{2}$ & $\frac{1}{2}$ & $\frac{1}{2}$ \\ \hline
d & \cellcolor{gray!100}$0$ & \cellcolor{gray!50}$1$ & \cellcolor{gray!100}$1$ & \cellcolor{gray!100}$0$ & \cellcolor{gray!50}$1$ & \cellcolor{gray!50}$1$ & \cellcolor{gray!50}$1$ \\ \hline
e & \cellcolor{gray!50}$0$ & $\frac{1}{2}$ & \cellcolor{gray!50}$0$ & \cellcolor{gray!50}$0$ & $\frac{1}{2}$ & $\frac{1}{2}$ & $\frac{1}{2}$ \\ \hline
f & \cellcolor{gray!50}$0$ & $\frac{1}{2}$ & \cellcolor{gray!50}$0$ & \cellcolor{gray!50}$0$ & $\frac{1}{2}$ & $\frac{1}{2}$ & $\frac{1}{2}$ \\ \hline
$ \delta $ & \cellcolor{gray!50}$0$ & $\frac{1}{2}$ & \cellcolor{gray!50}$0$ & \cellcolor{gray!50}$0$ & $\frac{1}{2}$ & $\frac{1}{2}$ & $\frac{1}{2}$ \\ \hline
\end{tabular}
\caption{An illustration of the payoff matrix for agent 1 (represented as the row player) under the kind of policy generated by the procedure from the proof of Theorem \ref{ckddg_upper_bound}. Here $ A = \{a, b, c, d, e, f\}, H = \{(1,b),(1,d),(2,a),(2,c),(2,d)\}, m_{b,a}=1,m_{b,c}=2,m_{b,d}=1,m_{d,a}=2,m_{d,c}=1,m_{d,d}=2 $. The highlighted actions $ H $ are depicted in gray. The available actions in a given scenario are not depicted.}
\end{table}

Let's create a procedure for randomly generating policies such that the average error of the policies generated by this procedure is at most our bound. Let's ``highlight'' some subset $ H $ of $ \{1, 2\} \times A $ where each element is chosen independently with probability $ \beta = \frac{\alpha}{n} $. For every $ a_1 \in A $ with $ (1, a_1) \in H $ and for every $ a_2 \in A $ with $ (2, a_2) \in H $, choose $ m_{a_1,a_2} \in \{1, 2\} $ uniformly at random. 

We can now define our policy $ M $: if the agents play actions $ a_1, a_2 $ with $ (1,a_1),(2,a_2) \in H $ then agent $ m_{a_1,a_2} $ wins. If only one agent plays a highlighted action then that agent wins, and if neither agent plays a highlighted action then both win with probability $ \frac{1}{2} $. Note that if both agents can access a highlighted action, the outcome will be determined by the sub-grid of highlighted actions.

To compute the average error of $ M $, we would naively generate a policy $ M $ as described, sample an anonymous scenario from $ S $, flip a coin to decide the order of the agents, and then compute the probability of the principal rewarding the wrong agent. Equivalently, though, we can sample an anonymous scenario and decide the ordering of the agents first, then highlight our actions, then decide which agent to favor if both play highlighted actions, then compute the probability of the principal rewarding the wrong agent. Let's consider those last two steps together: suppose we fix the set of highlighted actions that the desired winner can access (say, $ B_w \neq \emptyset $) and we do the same for the desired loser (denoted $ B_l \neq \emptyset $). Let's upper bound the probability that the desired loser wins by upper bounding the probability that the desired winner lacks a single highlighted action which guarantees a win against all of the desired loser's highlighted actions. The probability that any one of the desired winner's highlighted actions beats all of the desired loser's highlighted actions is $ 2^{-|B_l|} $, so the probability that every one of the desired winner's highlighted actions fails to guarantee a win is $ \left ( 1 - 2^{-|B_l|} \right )^{|B_w|} $. Note by inspection that this formula continues to hold if any of $ B_w, B_l $ are empty.

Let's now take another step back and consider the distribution of $ |B_w|, |B_l| $ for a fixed anonymous scenario $ s \in S $ and a fixed desired winner $ i \in \{1, 2\} $. By our procedure for choosing $ H $, $ |B_w| $ follows a binomial distribution with $ |C_w(s)| $ samples and a success probability of $ \beta $ while $ |B_l| $ follows a binomial distribution with $ |C_l(s)| $ samples and a success probability of $ \beta $. Furthermore, $ |B_w| $ and $ |B_l| $ are independent. Thus, conditional on being in anonymous scenario $ s $ (note that $ i $ doesn't matter), the probability of the desired loser winning is upper bounded by $ \underset{\substack{X \sim B(|C_w(s)|,\beta) \\ Y \sim B(|C_l(s)|,\beta)}}{\mathbb{E}} \left [ \left ( 1 - 2^{-Y} \right )^X \right ] $. Finally we can take the expectation over anonymous scenarios to see that some policy has error at most $ \underset{s \sim P}{\mathbb{E}} \left [ \underset{\substack{X \sim B(|C_w(s)|,\beta) \\ Y \sim B(|C_l(s)|,\beta)}}{\mathbb{E}} \left [ \left ( 1 - 2^{-Y} \right )^X \right ] \right ] $.

We know that $ |C_w(s)| < m $ or $ |C_l(s)| > n $ with probability at most $ \varepsilon $; thus, if we define $ S' = \{s \in S: |C_w(s)| \geq m, |C_l(s)| \leq n \} $ then our bound becomes $ \underset{s \in S'}{\sum} \left [ P(s) \underset{X,Y}{\mathbb{E}} \left [ \left ( 1 - 2^{-Y} \right )^X \right ] \right ] + \varepsilon $. Let's now focus on the inner expectation and for every $ i \in \{1, ..., k\}, j \in \{1, ..., l\} $ let's define the event $ E_{i,j} $ to say $ X \in [\mu_i \beta m, \mu_{i-1} \beta m) $ (or $ [\mu_1 \beta m, \infty) $ if $ i = 1 $) and to say $ Y \in (\nu_{j-1} \log_2(r), \nu_j \log_2(r)] $ (or $ (-\infty, \nu_1 \log_2(r)] $ if $ j = 1 $). Conditional on witnessing event $ E_{i,j} $ we have:
\begin{align*}
(1 - 2^{-Y})^X \leq e^{-X 2^{-Y}} \leq e^{-\mu_i \beta m 2^{-\nu_j \log_2(r)}} = \theta_{i,j}
\end{align*}
Thus our inner expectation is bounded by:
\begin{align*}
\sum_{i = 1}^{k} \sum_{j = 1}^{l} \theta_{i,j} \underset{X,Y}{\mathbb{P}} \left [ E_{i,j} \right ] + \underset{X,Y}{\mathbb{P}} \left [ \bigcap_{i,j} E_{i,j}^{c} \right ]
\end{align*}
A standard tail bounding argument shows that $ \mathbb{P}_{X,Y} \left [ E_{i,j} \right ] \leq \zeta_{i-1} \xi_{j-1} $ and that $ \mathbb{P}_{X,Y} \left [ \bigcap_{i,j} E_{i,j}^{c} \right ] \leq \zeta_k + \xi_l $ for all $ s \in S' $. For details, see appendix D. We can now pull the upper bound on our inner expectation outside the sum over $ s $. Since $ \sum_{s \in S'} P(s) \leq 1 $, we get our desired result. 
\end{proof}

\begin{corollary}
Suppose that $ (A, S, P, C_w, C_l) $ is a CKDDG and suppose that $ m, n \in \mathbb{N}^+ $ are such that $ |C_l(s)| \leq n, |C_w(s)| \geq m $ for all $ s \in S $. Then there is a policy whose error is bounded by a function of $ \frac{m}{n} $ which decays faster than the reciprocal of any polynomial.
\end{corollary}
\begin{proof}
Let's define $ r = \frac{m}{n} $ as in Theorem \ref{ckddg_upper_bound}. If we apply Theorem \ref{ckddg_upper_bound} with $ \alpha = 1, k = l = 1, \mu_1 = \nu_1 = \frac{1}{2} $ (note that we need $ r \geq 4 $ to ensure $ \nu_1 \geq \frac{\alpha}{\log_2(r)} $) and simplify then we get an upper bound of $ e^{-\frac{\sqrt{r}}{2}} + \left ( \frac{e}{2} \right )^{-\frac{r}{2}} + \frac{1}{e} r^{-\log_4 \left ( \frac{\log_4(r)}{e} \right )} $ which decays faster than the reciprocal of any polynomial.
\end{proof}

We now turn our attention to PIDDG's. Here we can again find good policies for broad classes of problems, though the error does not decay as dramatically as with CKDDG's. To understand the construction, suppose there is a scenario in a PIDDG where the desired loser has access to $ 100 $ actions and the desired winner has access to $ 10,000 $ actions. The principal can randomly rank all agent-action pairs and declare that whichever agent plays an action with a higher agent-action ranking wins. Whichever agent has an available action with the highest agent-action ranking among all available actions can guarantee themselves a win. The probability that the desired loser is this agent is $ \frac{100}{100 + 10,000} = \frac{1}{101} $. Note that this reasoning only requires the desired winner to know the ranking from the policy, not the desired loser's available actions. Let's now state and prove the general version:

\begin{theorem}
\label{piddg_upper_bound}
Any PIDDG $ (A, S, P, C_w, C_l) $ has a policy with error at most $ \underset{s \sim P}{\mathbb{E}} \left [ \begin{cases} \frac{|C_l(s)|}{|C_w(s)| + |C_l(s)|} & C_w(s) \cup C_l(s) \neq \emptyset \\ \frac{1}{2} & C_w(s) \cup C_l(s) = \emptyset \end{cases} \right ] $.
\end{theorem}
\begin{proof}
As before, let's consider a random family of policies and argue that the average error of our policies is at most our desired bound. Let's uniformly randomly select an injection $ \sigma : \{1, 2 \} \times A \to \{1, ..., 2|A|\} $ (in other words, rank the agent-action pairs). Define $ \sigma' : \{1, 2\} \times (A \cup \{\delta\}) \to \{0, ..., 2|A|\} $ via $ \sigma'(i,a) = \begin{cases} \sigma(i,a) & a \in A \\ 0 & a = \delta \end{cases} $. Let's define the policy $ M_{\sigma'} $ so that if agent $ 1 $ plays $ a_1 $ and if agent $ 2 $ plays $ a_2 $ then $ M_{\sigma'} $ rules in favor of whichever agent has a higher value of $ \sigma'(i,a_i) $ (and rewards both with probability $ \frac{1}{2} $ if $ \sigma'(1,a_1) = \sigma'(2,a_2) $, i.e. if $ a_1 = a_2 = \delta $). Note that the optimal behavior for an agent does not depend on the other agent's strategy or available actions.

To evaluate the average error of such a policy, let's do the same trick of exchanging the order of our expectations: first select an anonymous scenario $ s \in S $ and an ordering of the agents, then determine the probability that a random policy will reward the wrong agent. If neither agent has any non-trivial actions then the desired loser has a $ \frac{1}{2} $ chance of winning, otherwise there will be $ |C_w(s)| + |C_l(s)| $ valid agent-action pairs of which $ |C_l(s)| $ belong to the desired loser so there is a $ \frac{|C_l(s)|}{|C_w(s)| + |C_l(s)|} $ chance that the desired loser ends up with the highest-ranked agent-action pair. Thus the expected error of a policy in this family is exactly $ \underset{s \sim P}{\mathbb{E}} \left [ \begin{cases} \frac{|C_l(s)|}{|C_w(s)| + |C_l(s)|} & C_w(s) \cup C_l(s) \neq \emptyset \\ \frac{1}{2} & C_w(s) \cup C_l(s) = \emptyset \end{cases} \right ] $, showing that some policy has an error this low. 
\end{proof}

\begin{remark}
The bound from Theorem \ref{piddg_upper_bound} also holds for CKDDG's and the proof is identical. However, Theorem \ref{piddg_upper_bound} gives us an asymptotically weaker result than Theorem \ref{ckddg_upper_bound}. 
\end{remark}

We might hope that we could improve Theorem \ref{piddg_upper_bound} with a cleverer technique to approach the faster-than-reciprocal-of-any-polynomial decay of Theorem \ref{ckddg_upper_bound}. Unfortunately this is not possible in full generality: we can find a PIDDG such that the minimum error decays proportionally to the inverse of the ratio of the number of available actions for the desired winner to the number of available actions for the desired loser. 

\begin{theorem}
\label{piddg_lower_bound}
Suppose $ m, n \in \mathbb{N}^{+} $ with $ m \geq n $. Then there exists a PIDDG $ B = (A, S, P, C_w, C_l) $ such that $ |C_w(s)| \geq m, |C_l(s)| \leq n $ for all $ s \in S $ but $ e_B^{\text{PIDDG}}(M) \geq \frac{n}{2m} $ for all policies $ M $.
\end{theorem}
\begin{proof}
Let $ A = \{1, ..., m\} $ and $ S = \begin{pmatrix} A \\ n \end{pmatrix} $. Let's fix $ P $ to be uniform over $ S $ and let's define $ C_w(s) = A, C_l(s) = s $ for all $ s \in S $. Suppose $ M $ is a policy for this PIDDG.

$ e^{\text{PIDDG}}_B(M) $ is the average of two quantities: the probability that agent $ 1 $ wins while being the desired loser and the probability that agent $ 2 $ wins while being the desired loser. For the first case, suppose $ q_1^l : S \to (A \cup \{\delta\}) \to [0, 1] $ denotes the mixed strategy of agent $ 1 $ in a given scenario and suppose $ q_2^w : (A \cup \{\delta\}) \to [0, 1] $ denotes the mixed strategy of agent $ 2 $ -- note that agent $ 2 $ always gets every action and thus cannot infer anything about which scenario is unfolding. Suppose agent $ 1 $'s strategy is only supported on actions that they are permitted to use. Then the average win rate for agent $ 1 $ is $ \underset{s \sim P}{\mathbb{E}} \left [ \left ( q_1^l(s) \right )^{\intercal} M(1,\cdot,\cdot) q_2^w \right ] = \left ( \underset{s \sim P}{\mathbb{E}} \left [ q_1^l(s) \right ] \right )^{\intercal} M(1,\cdot,\cdot) q_2^w $. Likewise, if agent $ 2 $ is the desired loser then they will win with probability $ \left ( q_1^w \right )^{\intercal} M(2,\cdot,\cdot) \underset{s \sim P}{\mathbb{E}} \left [ q_2^l(s) \right ] $ where $ q_1^w, q_2^l $ are defined symmetrically to $ q_2^w, q_1^l $.

Now, let's construct some strategies for the desired losers to lower bound the error that they can guarantee. Momentarily pretend that we're in a new scenario where agents $ 1 $ and $ 2 $ have access to every possible action and define $ q_1^*, q_2^* : (A \cup \{\delta\}) \to [0, 1] $ to be resulting Nash equilibrium strategies. Now, for every $ i \in \{1, 2\} $ define $ q_i^l $ so that $ q_i^l(s,a) = q_i^*(a) $ if $ a \in s $, $ q_i^l(s,a) = 0 $ if $ a \in A \setminus s $, and $ q_i^l(s,\delta) = q_i^*(\delta) + \sum_{b \in A \setminus s} q_i^*(b) $. For every action $ a \in A $, exactly $ \frac{n}{m} $ of the scenarios include $ a $ so $ \underset{s \sim P}{\mathbb{E}} \left [ q_i^l(s,a) \right ] = \frac{n}{m} q_i^*(a) $. We clearly have $ \underset{s \sim P}{\mathbb{E}} \left [ q_i^l(s,\delta) \right ] \geq q_i^*(\delta) \geq \frac{n}{m} q_i^*(\delta) $ so:
\begin{align*}
\left ( \underset{s \sim P}{\mathbb{E}} \left [ q_1^l(s) \right ] \right )^{\intercal} M(1,\cdot,\cdot) q_2^w & \geq \left ( \frac{n}{m} q_1^* \right )^{\intercal} M(1,\cdot,\cdot) q_2^w \\
\left ( q_1^w \right )^{\intercal} M(2,\cdot,\cdot) \underset{s \sim P}{\mathbb{E}} \left [ q_2^l(s) \right ] & \geq \left ( q_1^w \right )^{\intercal} M(2,\cdot,\cdot) \frac{n}{m} q_2^*
\end{align*}
Now, since $ q_1^*, q_2^* $ are Nash equilibria we have:
\begin{align*}
\left ( q_1^* \right )^{\intercal} M(1,\cdot,\cdot) q_2^w & \geq \left ( q_1^* \right )^{\intercal} M(1,\cdot,\cdot) q_2^* \\
\left ( q_1^w \right )^{\intercal} M(2,\cdot,\cdot) q_2^* & \geq \left ( q_1^* \right )^{\intercal} M(2,\cdot,\cdot) q_2^*
\end{align*}
Thus the error of $ M $ is at least:
\begin{align*}
\frac{\frac{n}{m} \left ( q_1^* \right )^{\intercal} M(1,\cdot,\cdot) q_2^* + \frac{n}{m} \left ( q_1^* \right )^{\intercal} M(2,\cdot,\cdot) q_2^*}{2} = & \frac{n}{2m}
\end{align*}
\end{proof}

Although the details of this example are inconsequential, the proof technique generalizes: we may take a debate game, consider a new scenario which isn't present in the given debate game, examine the strategies that the agents would adopt if they were faced with this hypothetical scenario, and then understand the behavior of the agents in the real scenarios in terms of their hypothetical behavior in this hypothetical scenario.

\begin{remark}
We might ask whether we can attain a lower bound similar to Theorem \ref{piddg_lower_bound} in the CKDDG case. Indeed we can: see appendix E.
\end{remark}

\section{Future research}

Our results could be generalized in a number of directions. For example, many of our results apply to distinguishing debate games, but there are likely analogues in settings where the two agents do not have the same set of actions. One could also consider debates between three or more agents. Additionally, one could consider trying to optimize debate in a more iterated setting where the principal wishes to minimize the cost of running the debate protocol by terminating it as early as possible. Another natural question is whether there are any special cases of debate games where the principal can compute an optimal policy in polynomial time.

One particularly interesting direction involves looking at situations in which the principal cannot enforce whatever policy they wish. For example, suppose agents $ 1 $ and $ 2 $ are playing an extensive-form perfect information game where agent $ 1 $ either moves left or right and then agent $ 2 $ either moves left or right. 
For any fixed payoffs at the leaves of this extensive-form game, we can convert this extensive-form game into a normal-form game where agent $ 1 $ has two actions and agent $ 2 $ has four actions. In this case, the map from choices of utilities at the leaves of our extensive-form game to payoff matrices in the normal-form game is not surjective. For example, the principal cannot reward agent $ 1 $ if agent $ 1 $ goes left and agent $ 2 $ always goes left but also reward agent $ 2 $ if agent $ 1 $ goes left and agent $ 2 $ does whatever agent $ 1 $ did: agent $ 2 $ changed their action in the normal-form game (i.e. their pure strategy in the extensive-form game) but only in a way that differs on a branch that the principal never sees. Given that in many real-world games the principal cannot enforce whatever policy they wish, the question remains as to whether error bounds can be proven under such constraints.

Another particularly interesting direction involves examining debate games which are too big to specify. Are there natural ways to concisely represent debate games, especially in settings where the number of actions and/or scenarios is exponential in some other parameter? In this case, even our algorithms for {\em evaluating} a policy will require exponential time. What are good algorithms for reasoning about these representations? Let's draw particular attention to the plight of the agents when handed a compressed representation of a policy. When can the principal expect the agents to respond near-optimally? All of these questions merit further examination.

%\bibliographystyle{unsrt}  
%\bibliography{references}

\begin{thebibliography}{10}

\bibitem{Vickrey61}
William Vickrey.
\newblock Counterspeculation, auctions, and competitive sealed tenders.
\newblock {\em Journal of Finance}, 16:8--37, 1961.

\bibitem{Clarke71:Multipart}
Ed~H. Clarke.
\newblock Multipart pricing of public goods.
\newblock {\em Public Choice}, 11:17--33, 1971.

\bibitem{Groves73:Incentives}
Theodore Groves.
\newblock Incentives in teams.
\newblock {\em Econometrica}, 41:617--631, 1973.

\bibitem{Myerson81:Optimal}
Roger Myerson.
\newblock Optimal auction design.
\newblock {\em Mathematics of Operations Research}, 6:58--73, 1981.

\bibitem{Conitzer02:Mechanism}
Vincent Conitzer and Tuomas Sandholm.
\newblock Complexity of mechanism design.
\newblock In {\em Proceedings of the 18th Annual Conference on Uncertainty in Artificial Intelligence (UAI)}, pages 103--110, Edmonton, Canada, 2002.

\bibitem{Conitzer04:Self}
Vincent Conitzer and Tuomas Sandholm.
\newblock Self-interested automated mechanism design and implications for optimal combinatorial auctions.
\newblock In {\em Proceedings of the ACM Conference on Electronic Commerce (EC)}, pages 132--141, New York, NY, USA, 2004.

\bibitem{Green86:Partially}
Jerry Green and Jean-Jacques Laffont.
\newblock Partially verifiable information and mechanism design.
\newblock {\em Review of Economic Studies}, 53:447--456, 1986.

\bibitem{Yu11:Mechanism}
Lan Yu.
\newblock Mechanism design with partial verification and revelation principle.
\newblock {\em Autonomous Agents and Multi-Agent Systems}, 22(1):217--223, 2011.

\bibitem{Kephart21:Revelation}
Andrew Kephart and Vincent Conitzer.
\newblock The revelation principle for mechanism design with signaling costs.
\newblock {\em ACM Transactions on Economics and Computation (TEAC)}, 9(1):Article 6, 1--35, 2021.

\bibitem{Auletta11:Alternatives}
Vincenzo Auletta, Paolo Penna, Giuseppe Persiano, and Carmine Ventre.
\newblock Alternatives to truthfulness are hard to recognize.
\newblock {\em Autonomous Agents and Multi-Agent Systems}, 22(1):200--216, 2011.

\bibitem{Kephart15:Complexity}
Andrew Kephart and Vincent Conitzer.
\newblock Complexity of mechanism design with signaling costs.
\newblock In {\em Proceedings of the Fourteenth International Conference on Autonomous Agents and Multi-Agent Systems (AAMAS)}, pages 357--365, Istanbul, Turkey, 2015.

\bibitem{Zhang21:Automated}
Hanrui Zhang, Yu~Cheng, and Vincent Conitzer.
\newblock Automated mechanism design for classification with partial verification.
\newblock In {\em Proceedings of the Thirty-Fifth AAAI Conference on Artificial Intelligence}, pages 5789--5796, Virtual conference, 2021.

\bibitem{Irving18:AI}
Geoffrey Irving, Paul Christiano, and Dario Amodei.
\newblock {AI safety via debate}, 2018.
\newblock arXiv:1805.00899.

\bibitem{pmlr-v235-brown-cohen24a}
Jonah Brown-Cohen, Geoffrey Irving, and Georgios Piliouras.
\newblock Scalable {AI} safety via doubly-efficient debate.
\newblock In Ruslan Salakhutdinov, Zico Kolter, Katherine Heller, Adrian Weller, Nuria Oliver, Jonathan Scarlett, and Felix Berkenkamp, editors, {\em Proceedings of the 41st International Conference on Machine Learning}, volume 235 of {\em Proceedings of Machine Learning Research}, pages 4585--4602. PMLR, 21--27 Jul 2024.

\bibitem{Dung95:On}
P.~M. Dung.
\newblock On the acceptability of arguments and its fundamental role in nonmonotonic reasoning, logic programming and n-person games.
\newblock {\em Artificial Intelligence}, 77(2):321–358, 1995.

\bibitem{RahwanMechanism}
Iyad Rahwan and Kate Larson.
\newblock Mechanism design for abstract argumentation.
\newblock In {\em Proceedings of the Seventh International Conference on Autonomous Agents and Multi-Agent Systems (AAMAS)}, pages 1031--1039, Estoril, Portugal, 2008.

\bibitem{mitzenmacher05:Tails}
Michael Mitzenmacher and Eli Upfal.
\newblock {\em Probability and Computing: Randomized Algorithms and Probabilistic Analysis}.
\newblock Cambridge University Press, Cambridge, 2005.

\end{thebibliography}

\newpage
\appendix

\section{Proof Completions of Initial Observations}

\subsection{Proposition 3.5 Proof Completion}

To show more explicitly that $ B $ has a policy with zero error, consider the policy $ M : \{1, 2\} \times \{0, 1, 2, \delta\}^2 \to [0, 1] $ defined by:
\begin{align*}
M(i,a_1,a_2) & = \begin{cases} 1 & a_1, a_2 \in \{0, 1, 2\}, a_i = (a_{-i} + 1) \mod 3 \\ 0 & a_1, a_2 \in \{0, 1, 2\}, a_{i} = (a_{-i} - 1) \mod 3 \\ \frac{1}{2} & a_1, a_2 \in \{0, 1, 2\}, a_1 = a_2 \\ 1 & a_i \in \{0, 1, 2\}, a_{-i} = \delta \\ 0 & a_i = \delta, a_{-i} \in \{0, 1, 2\} \\ \frac{1}{2} & a_1 = a_2 = \delta \end{cases}
\end{align*}
Lastly, let's note that not only can we say that every policy for $ B' $ has strictly positive error, we furthermore know that every policy for $ B' $ has error at least $ \frac{1}{6} $. We get this lower bound from Theorem 5.5: $ B' $ is exactly the PIDDG constructed in the proof of Theorem 5.5 when $ m = 3, n = 1 $. 

\subsection{Proposition 3.6 Proof Completion}

To show that our CKDDG has a policy whose error is strictly less than the minimum error among policies which map to $ \{0, 1\} $, consider the policy $ M : \{1, 2\} \times \{0, 1, 2, 3, \delta\}^2 \to [0, 1] $ where $ M(1, \cdot, \cdot) $ is given by Table \ref{good_non_binary_policy}. $ M $ has zero error in anonymous scenarios $ 0, 1, 2 $. In anonymous scenario $ 3 $, the value for the desired loser is $ \frac{1}{2} $ (the desired loser can guarantee at least $ \frac{1}{2} $ by playing action $ 3 $ while the desired winner can guarantee at most $ \frac{1}{2} $ by uniformly playing $ \{0, 1, 2\} $). In anonymous scenario $ 4 $, the value for the desired loser is also $ \frac{1}{2} $ (the desired loser can guarantee at least $ \frac{1}{2} $ by uniformly mixing over $ \{0, 1, 2\} $ and the desired winner can guarantee at most $ \frac{1}{2} $ by playing action $ 3 $). Thus the error of $ M $ is $ \frac{1}{2} \cdot \frac{6}{100} + \frac{1}{2} \cdot \frac{4}{100} = 0.05 $.

\begin{table}
\centering
\def\arraystretch{1.25}

\begin{tabular}{c|c|c|c|c|c|}
& 0 & 1 & 2 & 3 & $ \delta $ \\ \hline
0 & $\frac{1}{2}$ & $ 0 $ & $ 1 $ & $ \frac{1}{2} $ & $1$ \\ \hline
1 & $1$ & $\frac{1}{2}$ & $0$ & $\frac{1}{2}$ & $1$ \\ \hline
2 & $0$ & $1$ & $\frac{1}{2}$ & $\frac{1}{2}$ & $1$ \\ \hline
3 & $\frac{1}{2}$ & $\frac{1}{2}$ & $\frac{1}{2}$ & $\frac{1}{2}$ & $1$ \\ \hline
$ \delta $ & $0$ & $0$ & $0$ & $0$ & $\frac{1}{2}$  \\ \hline
\end{tabular}
\caption{The payoffs for agent $ 1 $ (represented as the row player) under the policy $ M $ referenced in the proof of Proposition 3.6.}
\label{good_non_binary_policy}
\end{table}

Let's now show that every policy mapping to $ \{0, 1\} $ has an error of at least $ 0.056 $. To do this, suppose $ M_{\{0,1\}} : \{1,2\} \times \{0, 1, 2, 3, \delta\}^2 \to \{0, 1\} $ is arbitrary. Let $ \{1,2\} \times \{0, 1, 2, 3, 4\} $ be the de-anonymized scenarios in the CKDG induced by our CKDDG. 

I first claim that if $ M_{\{0,1\}} $ has error less than $ \frac{3}{40} $ then $ M_{\{0,1\}}(1, \cdot, \cdot) $ must take on the values listed in Table \ref{binary_policy_restriction}. To see why, note that if the error is less than $ \frac{3}{40} $ then for every de-anonymized scenario in $ \{1, 2\} \times \{0, 1, 2\} $ the desired winner must have an action to play which wins against both actions the desired loser can play (for anonymized scenarios $ 0, 1, 2 $, both sides get a single non-default action). If this were not the case then for both of the desired winner's actions, at least one action of the desired loser would favor the desired loser. The desired loser could then play each of their actions with $ \frac{1}{2} $ probability to guarantee an error of at least $ \frac{1}{2} \cdot \frac{3}{20} = \frac{3}{40} $.

\begin{table}
\centering
\def\arraystretch{1.25}

\begin{tabular}{c|c|c|c|c|c|}
& 0 & 1 & 2 & 3 & $ \delta $ \\ \hline
0 &  & $ 0 $ & $ 1 $ & & $1$ \\ \hline
1 & $1$ & & $0$ & & $1$ \\ \hline
2 & $0$ & $1$ & & & $1$ \\ \hline
3 & & & & & \\ \hline
$ \delta $ & $0$ & $0$ & $0$ & & \\ \hline
\end{tabular}
\caption{If $ M_{\{0,1\}} $ has error less than $ \frac{3}{40} $ then the payoffs for agent $ 1 $ (represented as the row player) must take on the values listed.}
\label{binary_policy_restriction}
\end{table}

Next, I claim that for every de-anonymized scenario, this action of the desired winner which beats both of the desired loser's actions must be the non-default action. To see why, again suppose towards a contradiction that in de-anonymized scenario $ (w,i) $ (for $ w \in \{1, 2\}, i \in \{0, 1, 2\} $), $ M_{\{0,1\}} $ favors agent $ w $ when agent $ w $ plays $ \delta $ and when agent $ -w $ plays either $ i $ or $ \delta $. Now, in the de-anonymized scenario $ (-w,(i-1) \mod 3) $, agent $ w $ (now the desired loser) can again play $ \delta $ to guarantee a win, thus resulting in an error of at least $ \frac{3}{20} $. Thus for de-anonymized scenarios $ \{1, 2\} \times \{0, 1, 2\} $, $ M_{\{0,1\}} $ must favor the desired winner whenever the desired winner plays their non-default action. Thus $ M_{\{0,1\}} $ must be of the form depicted in Table \ref{binary_policy_restriction} if the error of $ M_{\{0,1\}} $ is less than $ \frac{3}{40} $. 

Now, conditional on $ M_{\{0,1\}} $ having the form described in Table \ref{binary_policy_restriction}, I claim that the error of $ M_{\{0,1\}} $ is at least $ 0.056 $. To see this, let's partition the other de-anonymized scenarios into $ D_1 = \{(1,3),(2,4)\} $ and $ D_2 = \{(2,3),(1,4)\} $. I claim that the contribution to the error from $ D_1 $ is at least $ 0.028 $. To see this, let's case on the values of:
\begin{align*}
x & = M_{\{0,1\}}(1,\delta,3) \\
y & = \max(M_{\{0,1\}}(1,0,3),M_{\{0,1\}}(1,1,3),M_{\{0,1\}}(1,2,3))
\end{align*}
\begin{itemize}
\item If $ x = 0 $ and $ y = 0 $ then agent $ 2 $ can always win de-anonymized scenario $ (1,3) $ by playing action $ 3 $. Thus the error from $ D_1 $ is at least $ \frac{3}{100} $.
\item If $ x = 0 $ and $ y = 1 $ then agent $ 2 $ can win with probability at least $ \frac{1}{3} $ in de-anonymized scenario $ (1,3) $ by uniformly mixing over actions $ \{0, 1, 2\} $. Furthermore, agent $ 1 $ can guarantee a win outright in de-anonymized scenario $ (2,4) $ by playing one of the actions in $ \{0, 1, 2\} $ that caused $ y = 1 $. Thus the error from $ D_1 $ is at least $ \frac{1}{3} \cdot \frac{3}{100} + \frac{2}{100} = \frac{3}{100} $.
\item If $ x = 1 $ and $ y = 0 $ then agent $ 2 $ can win de-anonymized scenario $ (1,3) $ with probability at least $ \frac{3}{5} $ by playing actions $ 0,1,2 $ with probability $ \frac{1}{5} $ each and playing action $ 3 $ with probability $ \frac{2}{5} $. Meanwhile, agent $ 1 $ can win de-anonymized scenario $ (2,4) $ with probability at least $ \frac{1}{2} $ by playing action $ 0 $ with probability $ \frac{1}{2} $ and action $ \delta $ with probability $ \frac{1}{2} $. Thus the error from $ D_1 $ is at least $ \frac{3}{5} \cdot \frac{3}{100} + \frac{1}{2} \cdot \frac{2}{100} = \frac{28}{1000} $.
\item If $ x = 1 $ and $ y = 1 $ then the analysis is identical to the case where $ x = 0, y = 1 $ and the error from $ D_1 $ is at least $ \frac{3}{100} $.
\end{itemize}
We thus see that the error from $ D_1 $ is at least $ \frac{28}{1000} $ in all cases. An entirely symmetric analysis shows that the error from $ D_2 $ is also at least $ \frac{28}{1000} $. Thus, the overall error is at least $ \frac{56}{1000} $.

We have thus shown that if the error of $ M_{\{0,1\}} $ is less than $ \frac{3}{40} = \frac{75}{1000} $ then $ M_{\{0,1\}} $ must resemble Table \ref{binary_policy_restriction} and thus have error at least $ \frac{56}{1000} $, showing that no policy mapping to $ \{0,1\} $ can have error below $ 0.056 $.

\subsection{Proposition 3.7 Proof Completion}

To show that our CKDDG has a policy whose error is strictly less than the minimum error among policies which incentivize pure strategies, again consider the policy $ M $ described in Table \ref{good_non_binary_policy}. As before, scenarios $ 0, 1, 2 $ contribute no error while in scenario $ 3 $ the desired loser can win half of the time. Thus the error of $ M $ is again $ \frac{1}{2} \cdot \frac{1}{10} = 0.05 $. 

Let's show that every policy which incentivizes pure strategies has an error of at least $ 0.1 $. Suppose $ M_{\text{ps}} : \{1, 2\} \times \{0, 1, 2, 3, \delta\}^2 \to [0, 1] $ is a policy such that in every de-anonymized scenario there is a pure-strategy Nash equilibrium for the two agents. For each de-anonymized scenario $ (w,s) $ with $ w \in \{1, 2\}, s \in \{0, 1, 2, 3\} $, pick such a pure-strategy Nash equilibrium and define $ b_i(w,s) \in \{0, 1, 2, 3, \delta\} $ to be the action performed by agent $ i \in \{1, 2\} $ in this Nash equilibrium. 

Now, let's examine $ b_1(1,3) \in \{0, 1, 2, \delta\} $. We can always find some scenario $ s_1 \in \{0, 1, 2\} $ such that $ b_1(1,3) \in C_l(s_1) \cup \{\delta\} $ (if $ b_1(1,3) \neq \delta $ then $ s_1 = b_1(1,3) $ works, if $ b_1(1,3) = \delta $ then any $ s_1 $ will suffice). We have $ C_w(s_1) \cup \{\delta\} \subseteq \{0, 1, 2, 3, \delta\} = C_l(3) \cup \{\delta\} $ so the probability that agent $ 1 $ wins in de-anonymized scenario $ (1,3) $ is at most the probability that agent $ 1 $ wins in de-anonymized scenario $ (2,s_1) $ (because in de-anonymized scenario $ (2,s_1) $ agent $ 1 $ can again play action $ b_1(1,3) $ and agent $ 2 $'s ability to respond is no better). A symmetric argument shows that we can find $ s_2 \in \{0, 1, 2\} $ such that the probability that agent $ 2 $ wins in de-anonymized scenario $ (2,3) $ is at most the probability that agent $ 2 $ wins in de-anonymized scenario $ (1,s_2) $. 

Let $ p_1 \in [0, 1] $ denote the probability that agent $ 1 $ wins in de-anonymized scenario $ (1,3) $ and let $ p_2 \in [0, 1] $ denote the probability that agent $ 2 $ wins in de-anonymized scenario $ (2,3) $. We have thus demonstrated that the error of $ M_{\text{ps}} $ is at least:
\begin{align*}
&\, (1 - p_1) \cdot \frac{1}{20} + p_1 \cdot \frac{3}{20} + (1 - p_2) \cdot \frac{1}{20} + p_2 \cdot \frac{3}{20} \\
= &\, \frac{1}{10} + \frac{1}{10} p_1 + \frac{1}{10} p_2 \\
\geq &\, \frac{1}{10}
\end{align*}
Thus our result holds.

\section{Hardness Proof Details}

\subsection{Lemma 4.5 Proof Completion}

Let's show that if there is a satisfying assignment then our CKDG has a perfect policy, that if our CKDG has a perfect policy then there is a satisfying assignment, that $ C_1 $ is injective, and that $ C_2 $ is injective.

\begin{itemize}
\item To show that if there is a satisfying assignment then our CKDG has a perfect policy, suppose $ x_1, ..., x_n \in \{\bot,\top\} $ is a satisfying assignment for $ K_1 \land ... \land K_m $. Consider the policy defined by:
\begin{align*}
M(1,a_1,a_2) & = \begin{cases} 1 & \exists i,j[a_1 = (j,x_j) \land a_2 = (i,j)] \\ 0 & \text{otherwise} \end{cases}
\end{align*}
In the scenarios where we want agent $ 1 $ to win (indexed by variables), agent $ 1 $ can take the action specified by that variable's value in the satisfying assignment. In the scenarios where we want agent $ 2 $ to win (indexed by clauses), agent $ 2 $ can take the action corresponding to the literal rendering that clause true. Thus $ M $ is a perfect policy.
\item To show that if our CKDG has a perfect policy then there is a satisfying assignment, suppose $ M $ is a policy with zero error. Then for every scenario, there must be an action our desired winner can take such that the policy always favors them no matter the action by the desired loser (a \textit{perfect} action). Let $ b_1 : \{1, ..., n\} \to A_1 $ denote such a choice of perfect action for each of the scenarios in which agent $ 1 $ is the desired winner and let $ b_2 : \{1, ..., m\} \to A_2 $ denote the same for when agent $ 2 $ is the desired winner. For every $ j \in \{1, ..., n\} $ set $ x_j = \bot $ if $ b_1(j) = (j,\bot) $, $ x_j = \top $ if $ b_1(j) = (j,\top) $, and set $ x_j $ arbitrarily if $ b_1(j) = \delta $ (note that these are our only options). I claim that we have found a satisfying assignment for $ K_1 \land ... \land K_m $. To see why, let's examine $ K_i $ for $ i \in \{1, .., m\} $. We must have $ b_2(i) \in (\{i\} \times J_i) \cup \{\delta\} $. If $ b_2(i) = (i,j) $ then variable $ x_j $ must make $ K_i $ true -- we cannot have $ b_1(j) = \delta $ because then $ b_2(i) $ wouldn't be a perfect action so we must have $ b_1(j) = (j,x_j) $ so the only way for our policy to favor agent $ 1 $ on $ ((j,x_j),(i,j)) $ but agent $ 2 $ on $ ((j,\neg n_{i,j}),(i,j)) $ is if $ x_j \neq \neg n_{i,j} $ so $ x_j = n_{i,j} $ so $ K_i = \top $. Furthermore, if $ b_2(i) = \delta $ then every variable must render $ K_i $ true because if a single one didn't then again our perfect actions would conflict. Thus if $ M $ is a policy with zero error then our formula is satisfiable.
\item To see that $ C_1 $ is injective, let's start off by arguing that observing $ C_1(s) $ allows us to determine whether agent $ 1 $ is supposed to win in scenario $ s $. If agent $ 1 $ should win then their action set will only reference a single variable, but if agent $ 2 $ should win then their action set will reference at least two distinct variables. If agent $ 1 $ should win then we can determine the scenario by looking at which variable gets referenced. If agent $ 2 $ should win then our given action set allows us to determine the variables used in the clause and the state of their negations, which we've stated is enough to uniquely determine the clause. Thus $ C_1 $ is injective.
\item To see that $ C_2 $ is injective, again note that given $ C_2(s) $ we can determine which agent is the desired winner. If agent $ 1 $ is the desired winner then we will get exactly one action per clause but if agent $ 2 $ is the desired winner then we will get at least two actions per clause (clauses reference at least two variables and there is at least one clause). If agent $ 1 $ is the desired winner then we can look at which variable we're getting access to in each clause. If agent $ 2 $ is the desired winner then we can look at which clause all of our actions are coming from. Thus $ C_2 $ is also injective. 
\end{itemize}

\subsection{Lemma 4.6 Proof Completion}

We have four remaining tasks: showing that if the CKDG has a perfect policy then so does the CKDDG, showing that if the CKDDG has a perfect policy then so does the CKDG, and showing that $ C_w, C_l $ are injective.
\begin{itemize}
\item To show that if the CKDG has a perfect policy then so does the CKDDG, suppose $ M : \{1, 2\} \times (A_1 \cup \{\delta\}) \times (A_2 \cup \{\delta\}) \to [0, 1] $ is a perfect policy for our CKDG. Now, define $ M' : \{1, 2\} \times (A \cup \{\delta\})^2 \to [0, 1] $ so that for all $ a_1, b_1 \in A_1 \cup \{\delta\}, a_2, b_2 \in A_2 \cup \{\delta\} $:
\begin{align*}
M'(1,(1,a_1),(1,b_1)) & = \frac{1}{2} \\
M'(1,(1,a_1),(2,a_2)) & = M(1,a_1,a_2) \\
M'(1,(1,a_1),\delta) & = 1 \\
M'(1,(2,a_2),(1,a_1)) & = M(2,a_1,a_2) \\
M'(1,(2,a_2),(2,b_2)) & = \frac{1}{2} \\
M'(1,(2,a_2),\delta) & = 1 \\
M'(1,\delta,(1,a_1)) & = 0 \\
M'(1,\delta,(2,a_2)) & = 0 \\
M'(1,\delta,\delta) & = \frac{1}{2}
\end{align*}
I claim that $ M' $ is a perfect policy for our CKDDG. To show this, suppose $ s' \in S' $ and suppose $ w \in \{1, 2\} $ is the desired winner of a non-anonymous scenario in the unfurled version of our CKDDG. Let's define $ i_{s'} = \begin{cases} 1 & u(1,s') > u(2,s') \\ 2 & u(2,s') > u(1,s') \end{cases} $ (we cannot have equality since $ s' \in S' $). Since $ M $ is perfect and since $ s' $ occurred in our original CKDG with non-zero probability, agent $ i_{s'} $ must have had an action to play which beat every action agent $ -i_{s'} $ could play. Let's say that $ a(s') \in A_{i_{s'}} \cup \{\delta\} $ denotes this action. Now, if agent $ w $ performs action $ (i_{s'},a(s')) $ in our CKDDG, they are guaranteed to win our non-anonymized scenario (this can be seen by casework on $ w $ and $ i_{s'} $). Thus $ M' $ is a perfect policy for our CKDDG.
\item To show that if the CKDDG has a perfect policy then so does the CKDG, suppose $ M' : \{1, 2\} \times (A \cup \{\delta\})^2 \to [0, 1] $ is a perfect policy for our CKDDG. This means that for every non-anonymized scenario, the desired winner must be able to take an action which beats every available action of the desired loser. Let's say that $ b: \{1, 2\} \times S' \to A \cup \{\delta\} $ takes in a non-anonymized scenario and outputs such an action. 

For every $ s' \in S' $, let's define $ i_{s'} = \begin{cases} 1 & u(1,s') > u(2,s') \\ 2 & u(2,s') > u(1,s') \end{cases} $ as in the previous part. Now, I claim that there exists $ j \in \{1, 2\} $ such that for all $ s' \in S' $, if $ i_{s'} = j $, then $ b(j,s') \notin \{(j,\delta),\delta\} $ (note that it's very possible for both $ 1 $ and $ 2 $ to satisfy this property, in that case $ j $ is defined arbitrarily). To prove this statement, suppose towards a contradiction that there exists $ s_1' \in S' $ with $ i_{s_1'} = 1 $ and $ b(1,s_1') \in \{(1,\delta),\delta\} $ and also that there exists $ s_2' \in S' $ with $ i_{s_2'} = 2 $ and $ b(2,s_2') \in \{(2,\delta),\delta\} $. In the non-anonymized scenario $ (1, s_{1}') $ the desired loser (agent $ 2 $) has access to $ b(2,s_2') $ and in the non-anonymized scenario $ (2, s_{2}') $ the desired loser (agent 1) has access to $ b(1,s_1') $. Since these actions are perfect we must have $ M'(1,b(1,s_1'),b(2,s_2')) = 1 $ and $ M'(2,b(1,s_1'),b(2,s_2')) = 1 $, contradiction. Thus $ j \in \{1, 2\} $ exists.

Now we are ready to construct a perfect policy for our CKDG. Consider $ M : \{1, 2\} \times (A_1 \cup \{\delta\}) \times (A_2 \cup \{\delta\}) \to [0, 1] $ defined so that $ M(1, a_1, a_2) $ is equal to: 
\begin{align*}
\begin{cases} M'(1,(1,a_1),(2,a_2)) & a_1, a_2 \neq \delta \\ \max \left ( M'(1,(1,\delta),(2,a_2)), M'(1, \delta, (2,a_2)) \right ) & a_1 = \delta, a_2 \neq \delta \\ \min \left ( M'(1,(1,a_1),(2,\delta)), M'(1,(1,a_1),\delta) \right ) & a_1 \neq \delta, a_2 = \delta \\ 0 & a_1, a_2 = \delta, j = 1 \\ 1 & a_1, a_2 = \delta, j = 2 \end{cases}
\end{align*}
I claim that $ M $ is a perfect policy for our original CKDG. To show this, suppose we're in scenario $ s \in S $. If $ P(s) = 0 $ or $ u(1,s) = u(2,s) $ then $ s $ doesn't contribute anything to the overall error so we can assume $ s \in S' $. We want agent $ i_{s} $ to have a perfect action in the original CKDG. Let's define $ d(s) \in A_{i_s} \cup \{\delta\} $ by:
\begin{align*}
d(s) & = \begin{cases} a & b(i_s,s) = (i_s,a) \\ \delta & b(i_s,s) = \delta \end{cases}
\end{align*}
(Note that $ a \in A_{i_{s}} \cup \{\delta\} $.) Now $ d(s) $ is a perfect action for agent $ i_s $ in the original CKDG as can be seen by casework.
\item To show that $ C_w $ and $ C_l $ are injective, note that their outputs are always of the form $ \{i\} \times (C_i(s') \cup \{\delta\}) $ for some $ i \in \{1, 2\} $. Thus, if we see an output of the form $ (i, X) $ from either $ C_w $ or $ C_l $ then we can use injectivity of $ C_1, C_2 $ to find a unique $ s' \in S' $ with $ C_{i}(s') = X \setminus \{\delta\} $.
\end{itemize}

\section{Economic Example Parameter Setting}

We wish to find $ \alpha, k, l, (\mu_i)_{i=1}^{k}, (\nu_j)_{j=1}^{l} $ such that $ \sum_{i = 1}^{k} \sum_{j = 1}^{l} \theta_{i,j} \zeta_{i-1} \xi_{j-1} + \zeta_k + \xi_l \leq 0.000035 $ when $ r = \frac{3596}{7} $. Let's choose the value of $ \alpha $ as well as our sequences of $ \mu $'s and $ \nu $'s as hard-coded in the following Python 3 script: 

%\lstinputlisting[language=Python]{rectangular_error_new_notation.py}
\includegraphics[scale=0.8]{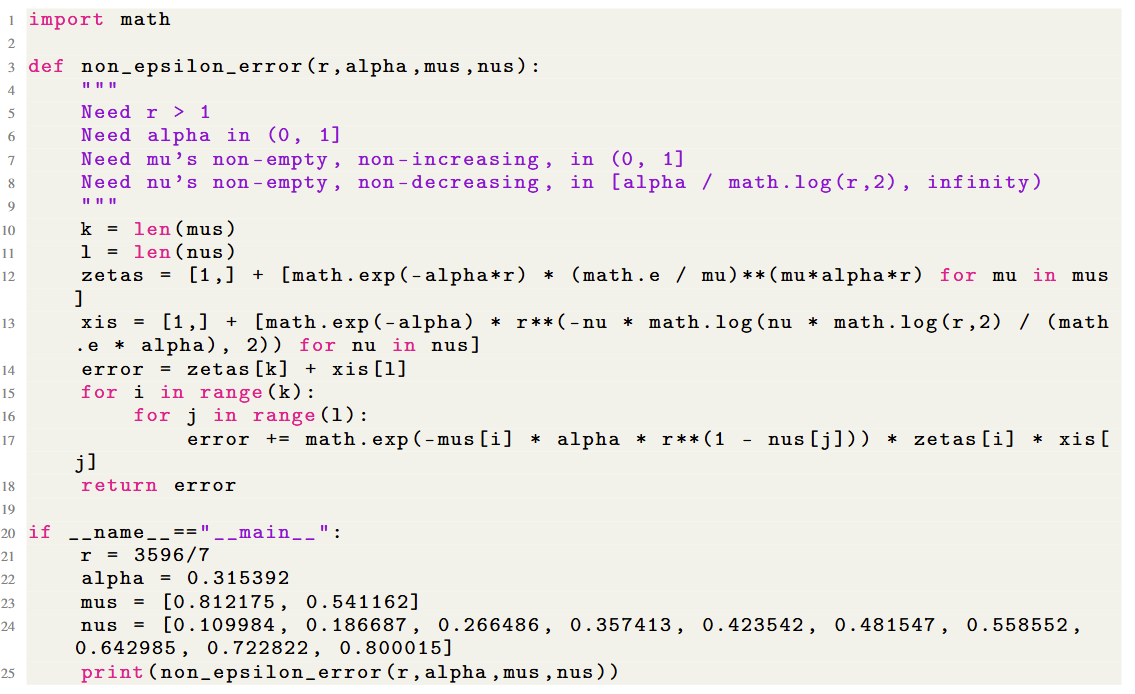}

If we run this code, we see that our error contribution with these parameter settings is 3.4722976180203346e-05 < 0.000035, as desired.

\section{Theorem 5.1 Proof Details}

To complete the proof of Theorem 5.1 we need upper bounds on the probabilities of our events $ E_{i,j} $ and we need an upper bound on the probability of avoiding the $ E_{i,j} $'s entirely. To facilitate both of these bounds, let's establish a lemma:
\begin{lemma}
\label{general_tail_bound}
Suppose $ s \in S' $.
\begin{enumerate}
\item[(i)] Suppose $ X $ is distributed according to $ B(|C_w(s)|,\beta) $ and suppose $ x \in \left ( 0, \beta m \right ] $. Then $ \mathbb{P} \left [ X \leq x \right ] \leq e^{-\beta m} \left ( \frac{e \beta m}{x} \right )^x $.
\item[(ii)] Suppose $ Y $ is distributed according to $ B(|C_l(s)|,\beta) $ and suppose $ y \in \left [ \beta n, \infty \right ) $. Then $ \mathbb{P} \left [ Y \geq y \right ] \leq e^{-\beta n} \left ( \frac{e \beta n}{y} \right )^y $.
\end{enumerate}
\end{lemma}
\begin{proof}
Let's use Chernoff bound on the tails of the binomial distribution:
\begin{enumerate}
\item[(i)] Theorem 4.5 part (1) from~\cite{mitzenmacher05:Tails} tells us that for all $ 0 < \gamma < 1 $:
\begin{align*}
\mathbb{P} \left [ X \leq (1 - \gamma) \beta |C_w(s)| \right ] & \leq \left ( \frac{e^{-\gamma}}{(1 - \gamma)^{(1 - \gamma)}} \right )^{\beta |C_w(s)|}
\end{align*}
Note that this formula continues to hold if $ \gamma = 0 $. We have $ \frac{e^{-\gamma}}{(1 - \gamma)^{(1 - \gamma)}} \leq 1 $ for all $ \gamma \in [0, 1) $ and we have $ \beta |C_w(s)| \geq \beta m $ since $ s \in S' $ so:
\begin{align*}
\mathbb{P} \left [ X \leq (1 - \gamma) \beta m \right ] & \leq \mathbb{P} \left [ X \leq (1 - \gamma) \beta |C_w(s)| \right ] \\
& \leq \left ( \frac{e^{-\gamma}}{(1 - \gamma)^{(1 - \gamma)}} \right )^{\beta |C_w(s)|} \\
& \leq \left ( \frac{e^{-\gamma}}{(1 - \gamma)^{(1 - \gamma)}} \right )^{\beta m}
\end{align*}
Now we can substitute $ \gamma = 1 - \frac{x}{\beta m} $ to see:
\begin{align*}
\mathbb{P} \left [ X \leq x \right ] & = \mathbb{P} \left [ X \leq \left ( 1 - \left ( 1 - \frac{x}{\beta m} \right ) \right ) \beta m \right ] \\
& \leq \left ( \frac{e^{-\left ( 1 - \frac{x}{\beta m} \right )}}{\left (1 - \left ( 1 - \frac{x}{\beta m} \right ) \right )^{\left ( 1 - \left ( 1 - \frac{x}{\beta m} \right ) \right )}} \right )^{\beta m} \\
& = \left ( \frac{e^{-\left ( 1 - \frac{x}{\beta m} \right )}}{\left ( \frac{x}{\beta m} \right )^{\left ( \frac{x}{\beta m} \right )}} \right )^{\beta m} \\
& = \frac{e^{-\left ( \beta m - x \right )}}{\left ( \frac{x}{\beta m} \right )^{x}} \\
& = e^{-\beta m} \left ( \frac{e \beta m}{x} \right )^x
\end{align*}
\item[(ii)] Theorem 4.4 part (1) from~\cite{mitzenmacher05:Tails} tells us that for all $ \gamma > 0 $:
\begin{align*}
\mathbb{P} \left [ Y \geq (1 + \gamma) \beta |C_l(s)| \right ] & \leq \left ( \frac{e^{\gamma}}{(1 + \gamma)^{(1 + \gamma)}} \right )^{\beta |C_l(s)|}
\end{align*}
Again, note that this formula continues to hold if $ \gamma = 0 $. If we substitute $ \gamma = \frac{y}{\beta |C_l(s)|} - 1 $ (note that $ y \geq \beta n \geq \beta |C_l(s)| $ since $ s \in S' $) then we see:
\begin{align*}
\mathbb{P} \left [ Y \geq y \right ] & = \mathbb{P} \left [ Y \geq \left ( 1 + \left ( \frac{y}{\beta |C_l(s)|} - 1 \right ) \right ) \beta |C_l(s)| \right ] \\
& \leq \left ( \frac{e^{\left ( \frac{y}{\beta |C_l(s)|} - 1 \right )}}{\left (1 + \left ( \frac{y}{\beta |C_l(s)|} - 1 \right ) \right )^{\left (1 + \left ( \frac{y}{\beta |C_l(s)|} - 1 \right ) \right )}} \right )^{\beta |C_l(s)|} \\
& = \frac{e^{\left ( y - \beta |C_l(s)| \right )}}{\left ( \frac{y}{\beta |C_l(s)|} \right )^{y}} \\
& = e^{-\beta |C_l(s)|} \left ( \frac{e \beta |C_l(s)|}{y} \right )^y
\end{align*}
We can note that for fixed $ w \geq 0 $, $ e^{-z} \left ( \frac{ez}{w} \right )^w $ is non-decreasing for $ z \in [0, w] $. Since $ \beta |C_l(s)|, \beta n $ both lie in $ [0, y] $ and $ \beta |C_l(s)| \leq \beta n $:
\begin{align*}
\mathbb{P} \left [ Y \geq y \right ] & \leq e^{-\beta |C_l(s)|} \left ( \frac{e \beta |C_l(s)|}{y} \right )^y \leq e^{-\beta n} \left ( \frac{e \beta n}{y} \right )^y
\end{align*}
\end{enumerate}
\end{proof}

We can now substitute specific values of $ x $ and $ y $ into these results:
\begin{lemma}
\label{specific_tail_bound}
Suppose $ s \in S' $.
\begin{enumerate}
\item[(i)] Suppose $ X $ has distribution $ B(|C_w(s)|,\beta) $ and suppose $ \mu \in (0, 1] $. Then $ \mathbb{P}[X \leq \mu \beta m] \leq e^{-\alpha r} \left ( \frac{e}{\mu} \right )^{\mu \alpha r} $.
\item[(ii)] Suppose $ Y $ has distribution $ B(|C_l(s)|,\beta) $ and suppose $ \nu \in \left [ \frac{\alpha}{\log_2(r)}, \infty \right ) $. Then $ \mathbb{P} \left [ Y \geq \nu \log_2(r) \right ] \leq e^{-\alpha} r^{-\nu \log_2 \left ( \frac{\nu \log_2(r)}{e \alpha} \right )} $.
\end{enumerate}
\end{lemma}
\begin{proof}
Let's appeal to Lemma \ref{general_tail_bound}.
\begin{enumerate}
\item[(i)] Define $ x = \mu \alpha r = \mu \beta m $. We have $ x \in (0, \beta m] $ since $ \mu \in (0, 1] $. By Lemma \ref{general_tail_bound} part (i):
\begin{align*}
\mathbb{P} \left [ X \leq \mu \beta m \right ] & \leq e^{-\beta m} \left ( \frac{e \beta m}{x} \right )^{x} = e^{-\alpha r} \left ( \frac{e \alpha r}{\mu \alpha r} \right )^{\mu \alpha r} = e^{-\alpha r} \left ( \frac{e}{\mu} \right )^{\mu \alpha r}
\end{align*}
\item[(ii)] Define $ y = \nu \log_2(r) $. We have $ y \in \left [ \beta n, \infty \right ) $ because:
\begin{align*}
\nu \log_2(r) \geq \frac{\alpha}{\log_2(r)} \log_2(r) = \alpha = \beta n 
\end{align*}
Thus:
\begin{align*}
\mathbb{P} \left [ Y \geq \nu \log_2(r) \right ] & \leq e^{-\beta n} \left ( \frac{e \beta n}{y} \right )^y \\
& = e^{-\alpha} \left ( \frac{e \alpha}{\nu \log_2(r)} \right )^{\nu \log_2(r)} \\
& = e^{-\alpha} 2^{\log_2 \left ( \frac{e \alpha}{\nu \log_2(r)} \right ) \nu \log_2(r)} \\
& = e^{-\alpha} 2^{- \nu \log_2(r) \log_2 \left ( \frac{\nu \log_2(r)}{e \alpha} \right )} \\
& = e^{-\alpha} r^{- \nu \log_2 \left ( \frac{\nu \log_2(r)}{e \alpha} \right )}
\end{align*}
\end{enumerate}
\end{proof}

Now we can proceed with proving our upper bounds for Theorem 5.1:
\begin{itemize}
\item Suppose $ s \in S', X \sim B(|C_w(s)|,\beta), Y \sim B(|C_l(s)|,\beta) $ where $ X $ and $ Y $ are independent. Suppose $ i \in \{1, ..., k\}, j \in \{1, ..., l\} $. We have:
\begin{align*}
&\, \mathbb{P} \left [ E_{i,j} \right ] \\
= &\, \mathbb{P} \Bigg [ X \in \left [ \mu_i \beta m, \begin{cases} \mu_{i - 1} \beta m & i \geq 2 \\ \infty & i = 1 \end{cases} \right ), \\
&\, Y \in \left ( \begin{cases} \nu_{j - 1} \log_2(r) & j \geq 2 \\ -\infty & j = 1 \end{cases}, \nu_j \log_2(r) \right ] \Bigg ] \\
\leq &\, \mathbb{P} \left [ X < \begin{cases} \mu_{i - 1} \beta m & i \geq 2 \\ \infty & i = 1 \end{cases}, Y > \begin{cases} \nu_{j - 1} \log_2(r) & j \geq 2 \\ -\infty & j = 1 \end{cases} \right ] \\
= &\, \mathbb{P} \left [ X < \begin{cases} \mu_{i - 1} \beta m & i \geq 2 \\ \infty & i = 1 \end{cases} \right ] \mathbb{P} \left [ Y > \begin{cases} \nu_{j - 1} \log_2(r) & j \geq 2 \\ -\infty & j = 1 \end{cases} \right ] \\
\leq &\, \left ( \begin{cases} e^{-\alpha r} \left ( \frac{e}{\mu_{i-1}} \right )^{\mu_{i - 1} \alpha r} & i \geq 2 \\ 1 & i = 1 \end{cases} \right ) \left ( \begin{cases} e^{-\alpha} r^{-\nu_{j-1} \log_2 \left ( \frac{\nu_{j-1} \log_2(r)}{e \alpha} \right )} & j \geq 2 \\ 1 & j = 1 \end{cases} \right ) \\
= &\, \zeta_{i-1} \xi_{j-1}
\end{align*}
\item Suppose $ s \in S' $. The only way for $ X, Y $ to avoid all of the $ E_{i,j} $ events is if $ X < \mu_k \beta m $ or if $ Y > \nu_l \log_2(r) $. We can thus apply a union bound to see that:
\begin{align*}
\mathbb{P} \left [ \underset{\substack{i \in \{1, ..., k\} \\ j \in \{1, ..., l\}}}{\bigcap} E_{i,j}^c \right ] & \leq \mathbb{P} \left [ X < \mu_k \beta m \right ] + \mathbb{P} \left [ Y > \nu_l \log_2(r) \right ] \\
& \leq e^{-\alpha r} \left ( \frac{e}{\mu_k} \right )^{\mu_k \alpha r} + e^{-\alpha} r^{-\nu_l \log_2 \left ( \frac{\nu_l \log_2(r)}{e \alpha} \right )} \\
& = \zeta_k + \xi_l
\end{align*}
\end{itemize}

Thus the proof of Theorem 5.1 is complete.

\section{Common Knowledge Error Lower Bound}

\begin{theorem}
Suppose $ m, n \in \mathbb{N}^+ $ with $ m \geq n $. Then there exists a CKDDG $ B = (A, S, P, C_w, C_l) $ such that $ |C_w(s)| \geq m, |C_l(s)| \leq n $ for all $ s \in S $ but $ e^{\text{CKDDG}}_{B}(M) \geq \frac{2n - m}{2m} $ for all policies $ M $.
\end{theorem}
\begin{proof}
Let's consider the CKDDG with the exact same specification as the PIDDG from the proof of Theorem 5.5: $ A = \{1, ..., m\} $, $ S = \begin{pmatrix} A \\ n \end{pmatrix} $, $ P $ is uniform over $ S $, and $ C_w(s) = A, C_l(s) = s $ for all $ s \in S $. Let $ M $ be a policy for this CKDDG.

Let's repeat the key idea from the proof of Theorem 5.5: let's pretend that we're in a new de-anonymized scenario where agents $ 1 $ and $ 2 $ have access to every possible action and let's define $ q_1^*, q_2^* : (A \cup \{\delta\}) \to [0, 1] $ to be Nash equilibrium strategies in this made-up de-anonymized scenario. Let $ v_i $ denote the value to agent $ i \in \{1, 2\} $ of this de-anonymized scenario where $ v_1 + v_2 = 1 $. Let's again consider the strategies $ q_i^l : S \to (A \cup \{\delta\}) \to [0, 1] $ for $ i \in \{1, 2\} $ defined by:
\begin{align*}
q_i^l(s)(a) & = \begin{cases} q_i^*(a) & a \in s \\ 0 & a \in A \setminus s \\ q_i^*(\delta) + \sum_{b \in A \setminus s} q_i^*(b) & a = \delta \end{cases}
\end{align*}
Let's lower-bound the error by considering the value agent $ i \in \{1, 2\} $ receives in de-anonymized scenario $ (-i,s) $ by committing to the strategy $ q_i^l(s) $ for $ i \in \{1, 2\}, s \in S $. If agent $ 1 $ commits to the strategy $ q_1^l(s) $ in de-anonymized scenario $ (2, s) $ then agent $ 1 $ can guarantee a winning probability of at least:
\begin{align*}
&\, \min_{a_2 \in A \cup \{\delta\}} (q_1^l(s))^{\intercal} M(1,\cdot,a_2) \\
= &\, \min_{a_2 \in A \cup \{\delta\}} \left ( \sum_{a_1 \in s} q_1^*(a_1) M(1,a_1,a_2) + \left ( q_1^*(\delta) + \sum_{b \in A \setminus s} q_1^*(b) \right ) M(1,\delta,a_2) \right ) \\
\geq &\, \min_{a_2 \in A \cup \{\delta\}} \left ( \sum_{a_1 \in s} q_1^*(a_1) M(1,a_1,a_2) + q_1^*(\delta) M(1,\delta,a_2) \right ) \\
= &\, \min_{a_2 \in A \cup \{\delta\}} \left ( (q_1^*)^{\intercal} M(1,\cdot,a_2) - \sum_{a_1 \in A \setminus s} q_1^*(a_1) M(1,a_1,a_2) \right ) \\
\geq &\, \min_{a_2 \in A \cup \{\delta\}} \left ( (q_1^*)^{\intercal} M(1,\cdot,a_2) - \sum_{a_1 \in A \setminus s} q_1^*(a_1) \right ) \\
= &\, \min_{a_2 \in A \cup \{\delta\}} (q_1^*)^{\intercal} M(1,\cdot,a_2) - \sum_{a_1 \in A \setminus s} q_1^*(a_1) \\
= &\, v_1 - \sum_{a_1 \in A \setminus s} q_1^*(a_1)
\end{align*}
In the last step we have used the fact that $ q_1^* $ is a Nash equilibrium strategy. A symmetric line of reasoning shows that if agent $ 2 $ commits to the strategy $ q_2^l(s) $ in de-anonymized scenario $ (1,s) $ then agent $ 2 $ can guarantee a winning probability of at least $ v_2 - \sum_{a_2 \in A \setminus s} q_2^*(a_2) $.

Finally, taking the expectation over $ s \in S $ shows that:
\begin{align*}
e_B^{\text{CKDDG}}(M) \geq &\, \underset{s \sim P}{\mathbb{E}} \left [ \frac{\left ( v_1 - \sum_{a_1 \in A \setminus s} q_1^*(a_1) \right ) + \left ( v_2 - \sum_{a_2 \in A \setminus s} q_2^*(a_2) \right )}{2} \right ] \\
= &\, \frac{v_1 + v_2}{2} - \frac{1}{2} \underset{s \sim P}{\mathbb{E}} \left [ \sum_{a_1 \in A \setminus s} q_1^*(a_1) \right ] - \frac{1}{2} \underset{s \sim P}{\mathbb{E}} \left [ \sum_{a_2 \in A \setminus s} q_2^*(a_2) \right ] \\
= &\, \frac{1}{2} - \frac{1}{2} \underset{s \sim P}{\mathbb{E}} \left [ \sum_{a_1 \in A} \chi_{A \setminus s}(a_1) q_1^*(a_1) \right ] - \frac{1}{2} \underset{s \sim P}{\mathbb{E}} \left [ \sum_{a_2 \in A} \chi_{A \setminus s}(a_2) q_2^*(a_2) \right ] \\
= &\, \frac{1}{2} - \frac{1}{2} \sum_{a_1 \in A} q_1^*(a_1) \underset{s \sim P}{\mathbb{E}} \left [ \chi_{A \setminus s}(a_1) \right ] - \frac{1}{2} \sum_{a_2 \in A} q_2^*(a_2) \underset{s \sim P}{\mathbb{E}} \left [ \chi_{A \setminus s}(a_2) \right ] \\
= &\, \frac{1}{2} - \frac{1}{2} \sum_{a_1 \in A} q_1^*(a_1) \frac{m - n}{m} - \frac{1}{2} \sum_{a_2 \in A} q_2^*(a_2) \frac{m - n}{m} \\
\geq &\, \frac{1}{2} - \frac{1}{2} \frac{m - n}{m} - \frac{1}{2} \frac{m - n}{m} \\
= &\, \frac{2n - m}{2m}
\end{align*}
\end{proof}

\end{document}